\documentclass[11pt]{article}
\usepackage{fullpage}
\usepackage{latexsym}
\usepackage{amsthm}
\usepackage{amsfonts}
\usepackage{amssymb}
\usepackage{amsmath}
\usepackage{nicefrac}
\usepackage{authblk}
\usepackage{float}
\usepackage[backref=page,colorlinks,citecolor=blue,linkcolor=purple]{hyperref}
\usepackage{graphicx}
\usepackage{xcolor}
\usepackage{subfigure}
\usepackage{algorithm}
\usepackage{algorithmic}
\usepackage{import}
\usepackage{url}

\allowdisplaybreaks[1]

\def\RTMutualInfLemma{4.5}
\def\RTCovarMlInfLemma{5.2}
\def\RTAlphaIndepThm{4.6}
\def\RTVarianceThm{5.6}
\def\RTBalanceCor{5.7}
\def\R{\mathbb{R}}

\def\eps{\epsilon}

\def\dif{\:\mbox{d}}

\def\iprod#1#2{\langle #1, #2 \rangle}
\DeclareMathOperator*{\expectationE}{\mathbb{E}}
\DeclareMathOperator*{\VarO}{Var}

\DeclareMathOperator{\Cov}{Cov}

\DeclareMathOperator{\poly}{poly}

\newcommand{\Var}[2][]{\VarO_{#1}\!\left[#2\right]}
\newcommand{\Exp}[2][]{\expectationE_{#1}\!\left[#2\right]}
\newcommand{\pro}[2][]{\Pr_{#1}\!\left[#2\right]}

\def\sign{\mathrm{sign}}

\def\csp#1{\textsc{Max CSP}$(#1)$}
\def\balcsp#1{\textsc{Max Bisect-CSP}$(#1)$}
\def\balcsps{\textsc{Max Bisect-CSP}}
\def\maxcut{\textsc{Max Cut}}
\def\maxbisection{\textsc{Max Bisection}}
\def\maxtwosat{\textsc{Max 2-Sat}}
\def\balmaxtwosat{\textsc{Max Bisect-2-Sat}}

\def\bu{{\bf u}}
\def\bv{{\bf v}}
\def\bw{{\bf w}}
\def\obw{\overline{\bw}}
\def\boldg{{\bf g}}

\def\set#1{\left\{#1\right\}}
\def\symd{\bigtriangleup}

\newcommand{\norm}[2][]{\left\|#2\right\|_{#1}}
\newtheorem{theorem}{Theorem}[section]

\newtheorem{lemma}[theorem]{Lemma}

\newtheorem{conjecture}[theorem]{Conjecture}
\newtheorem{definition}[theorem]{Definition}
\newtheorem{fact}[theorem]{Fact}

\newtheorem{numericalclaim}[theorem]{Claim (Numerical)}

\newenvironment{relemma}[1]{\vspace{0.25cm}

\noindent{\bf Lemma \ref{#1} (restated)}.\it}{\vspace{0.25cm}%
}

\floatstyle{boxed}
\newfloat{algo}{hbt}{alg}
\floatname{algo}{Algorithm}

\newsavebox{\sdpbox}

\def\ugc{Unique Games Conjecture}

\newcommand{\deriv}{\partial}
\newcommand{\ud}{\text{d}}
\newcommand{\trho}{\tilde{\rho}}
\DeclareMathOperator{\Val}{\mathsf{Val}}
\DeclareMathOperator{\Opt}{\mathsf{Opt}}
\DeclareMathOperator{\SDPVal}{\mathsf{SDPVal}}
\DeclareMathOperator{\SDP}{\mathsf{SDP}}
\newcommand{\Conf}{\mathbf{Conf}}

\DeclareMathOperator{\argmin}{arg min}
\DeclareMathOperator{\argmax}{arg max}
\DeclareMathOperator{\sgn}{sgn}

\newcommand{\la}{\textsf{Lasserre}}
\newcommand{\iSDP}{\textsf{SDP}}

\newcommand{\alphagw}{\alpha_{\textnormal{GW}}}
\newcommand{\alphagwnum}{0.8786}

\newcommand{\alphalinear}{0.8736}
\newcommand{\alphalinearfull}{0.87368287}
\newcommand{\alphalinearbestc}{0.86450318}
\newcommand{\worstconfiglinearb}{0.176945}
\newcommand{\worstconfiglinearrho}{-0.646110}
\newcommand{\alphalinearlimit}{0.873829}
\newcommand{\alphalinearlimitp}{0.931935}

\newcommand{\alphasecond}{0.8776}
\newcommand{\alphasecondfull}{0.87765366}

\newcommand{\alphabest}{0.8776}

\newcommand{\RTratio}{0.85}

\title{Better Balance by Being Biased:\\A $\alphabest$-Approximation for Max Bisection}

\author[*]{Per Austrin}
\author[*]{Siavosh Benabbas}
\author[$\dagger$]{Konstantinos Georgiou}
\affil[*]{Department of Computer Science, University of Toronto}
\affil[$\dagger$]{Department of Combinatorics \& Optimization, University of Waterloo}
\affil[$ $]{\texttt{\{austrin,siavosh\}@cs.toronto.edu}, \texttt{k2georgiou@math.uwaterloo.ca}}

\begin{document}
\maketitle
\begin{abstract}
  Recently Raghavendra and Tan (SODA 2012) gave a
  $\RTratio$-approximation algorithm for the \maxbisection{} problem.
  We improve their algorithm to a \alphabest{}-approximation.  As
  \maxbisection{} is hard to approximate within $\alphagw + \epsilon
  \approx \alphagwnum$ under the Unique Games Conjecture (UGC), our
  algorithm is nearly optimal.  We conjecture that \maxbisection{} is
  approximable within $\alpha_{GW}-\epsilon$, i.e., that the bisection
  constraint (essentially) does not make \maxcut{} harder.

  We also obtain an optimal algorithm (assuming the UGC) for the
  analogous variant of \maxtwosat{}.  Our approximation ratio for this
  problem exactly matches the optimal approximation ratio for
  \maxtwosat{}, i.e., $\alpha_{LLZ} + \epsilon \approx 0.9401$,
  showing that the bisection constraint does not make \maxtwosat{}
  harder. This improves on a $0.93$-approximation for this problem due
  to Raghavendra and Tan.
\end{abstract}

\newpage
\tableofcontents

\newpage
\section{Introduction}

In the \maxbisection{} problem we are given a (weighted) graph $G =
(V,E)$, and the objective is to find a \emph{bisection} $V = S \cup
\overline{S}$, $|S| = |\overline{S}| = |V|/2$ such that the number
(weight) of edges between $S$ and $\overline{S}$ is maximized.

\maxbisection{} is closely related to the \maxcut{} problem, in which
the constraint $|S| = |\overline{S}|$ is dropped.  \maxcut{} is one of
Karp's original 21 NP-Complete problems~\cite{karp72reducibility} and
is one of the most well-studied NP-hard problems.  In a seminal work
Goemans and Williamson~\cite{goemans95improved} show how to use
Semidefinite Programming to obtain an $\alpha_{GW}\approx \alphagwnum$
approximation algorithm for \maxcut{}. Here we say that a (randomized) algorithm
is an $\alpha$-approximation if for every graph $G$ it outputs a cut
in which the number of edges cut is (in expectation) at least an $\alpha$ fraction of
the optimum number of edges cut.
Since then, a series of results have continued the study of
the approximability of \maxcut{}, by providing improved approximation
ratios in special classes of graphs \cite{arora99polynomial,feige02improved}, integrality gaps for
(strengthenings of) the Semidefinite Programming relaxation
\cite{feige01integrality,khot05integrality,khot09sdp}, and hardness of
approximation results \cite{hastad01optimal}.  In a celebrated result,
Khot et al.~\cite{khot07optimal} proved that, assuming the
\emph{Unique Games Conjecture}, it is hard to approximate \maxcut{}
within a factor $\alpha_{GW} + \epsilon$ for any $\epsilon > 0$.
Subsequently, O'Donnell and Wu \cite{odonnell08optimal} determined the
entire ``approximability curve'' of \maxcut{}, thereby completely
settling the approximability of \maxcut{} modulo the Unique Games
Conjecture.

Overall, one can think of \maxcut{} as a problem whose
approximability has been (essentially) resolved. It is worthwhile to
note that this mostly stems from the \emph{local} nature of
the problem, i.e., that one can analyze the value of the objective
function by analyzing whether each edge is cut separately. 
In other words both feasibility and the objective value
of a potential solution to \maxcut{} are very local.  

\maxbisection{} on the other hand has a global condition $|S| =
|\overline{S}|$ determining feasibility.  It is perhaps not surprising
then that settling the approximability of \maxbisection{} has turned
out to be more challenging.  While it is well-known and easy
to see that \maxbisection{} is at least as hard to approximate as
\maxcut{} (the reduction from \maxcut{} to \maxbisection{} simply
outputs two disjoint copies of the graph), it is not known
whether the converse holds, i.e.,
\begin{center}
  \emph{Is \maxbisection{} as easy to approximate as \maxcut{}?}
\end{center}
There has been a long chain of results obtaining improved
approximation algorithms for \maxbisection{}.  Frieze and
Jerrum~\cite{frieze97improved}, in the first nontrivial approximation
algorithm, showed that the problem can be
approximated to within a factor of $0.6514$. Subsequently,
Ye~\cite{ye01approximation}, Halperin and
Zwick~\cite{halperin02unified}, and Feige and
Langberg~\cite{feige06rprr} gave algorithms for \maxbisection{} with
ratios $0.699$, $0.7016$, and $0.7028$ respectively. For the case of
regular graphs, Feige et al.~\cite{feige01note} showed that one can
improve the approximation ratio to $0.795$ (or even $0.834$ for
3-regular graphs). Very recently, in a 
significant improvement,
Raghavendra and Tan~\cite{raghavendra12approximating} gave a
$\RTratio$-approximation algorithm (based on a computer-assisted
analysis), improving upon these previous results.


\subsection{Our Contributions}

Our main contribution is a further improvement on the approximability
of \maxbisection{}.  We present a new approximation algorithm for
\maxbisection{} with approximation factor $\alpha$, where $\alpha$ is
the minimum of a certain function over a simple 3-dimensional
polytope.  Using a Matlab program we non-rigorously estimate that
$\alpha \approx \alphasecondfull$, and using a computer-assisted case
analysis we can formally prove this up to four digits of accuracy.

\begin{theorem}\label{thm:main}
  \maxbisection{} is approximable in polynomial time to within a factor $\alphabest$.
\end{theorem}
\noindent
As mentioned above, \maxbisection{} is as hard as \maxcut{}, and hence the UGC
implies that \maxbisection{} cannot be approximated
to within a factor $\alphagw + \epsilon \approx \alphagwnum$ for any $\epsilon>0$, so our
approximation ratio is off from the optimal by less than $10^{-3}$.
As it turns out, our algorithm has a lot of flexibility, indicating
that further improvements may be possible.  We remark that, while
polynomial, the running time of the algorithm is somewhat abysmal;
loose estimates places it somewhere around
$O\!\left(n^{10^{100}}\right)$; the running time of the algorithm of \cite{raghavendra12approximating} is similar.

One can consider bisection-like variants of any \textsc{Max CSP}.  We
refer to the resulting problem as \balcsps.  For
instance, in the \balmaxtwosat{} problem, we are given a \maxtwosat{}
instance and the goal is to obtain an assignment to the variables
maximizing the number of satisfied clauses, subject to the constraint
that exactly half of the variables are set to true, and the other half
are set to false.  For \balmaxtwosat{},
\cite{raghavendra12approximating} gave a $0.93$-approximation
algorithm (again based on a computer-assisted analysis).  Under the
Unique Games Conjecture, the approximation threshold for \maxtwosat{}
is known to be $\alpha_{LLZ} \approx 0.9401$ \cite{lewin02improved, austrin07balanced}
and again it is easy to prove that \balmaxtwosat{} can not be easier
than this (see Section~\ref{sec:2sat}).  We show that a simple
modification to the algorithm of \cite{raghavendra12approximating}
yields the optimal approximation ratio $\alpha_{LLZ}$ for
\balmaxtwosat{}.

\begin{theorem}\label{thm:2sat}
  For every $\epsilon > 0$, \balmaxtwosat{} can be approximated to within $\alpha_{LLZ} -
  \epsilon$ in time $n^{\poly(1/\epsilon)}$. Here $\alpha_{LLZ} \approx 0.9401$ is the approximation
  threshold for \maxtwosat{}.
\end{theorem}

This may seem surprising at first, but boils down to what seems to be
a lucky coincidence: the rounding scheme of
\cite{raghavendra12approximating} for the semidefinite program uses a
certain variant of random hyperplane rounding.  We generalize this to
a certain family of random hyperplane-based roundings, and it turns
out that the optimal rounding scheme for \maxtwosat{} already
comes from this family.

Given these results, we think it is likely that \maxbisection{} is
essentially as easy to approximate as \maxcut{}, and make the
following conjecture.

\begin{conjecture}
For every $\epsilon > 0$, $\maxbisection{}$ is approximable in
polynomial time within a factor $\alpha_{GW}-\epsilon$.
\end{conjecture}

\subsection{Techniques and Comparison to Previous Work}

All approximation algorithms for \maxbisection{} to date use a
semidefinite programming relaxation similar to the Goemans-Williamson
algorithm for \maxcut{}.  In its standard form, each vertex $i$ of the
graph is associated with a high-dimensional unit vector $\bv_i$ simulating
the integral values $\pm 1$, and the goal is to choose these vectors in
such a way that pairs of vertices connected by edges are as far apart as
possible. To be more concrete the goal is to maximize the ``objective value''
of the vectors defined as $\sum_{ij\in E} (1-\iprod{\bv_i}{\bv_j})/2$.  
There is also an additional balance constraint encoding that the vectors somehow correspond to a bisection as opposed to an arbitrary cut (this balance constraint is not important for the high-level discussion of this section).
An (essentially) optimal set of such vectors can be found in polynomial time, and the next step is
to ``round'' these vectors to a bisection of the vertices. 

The vast majority of SDP-based approximation algorithms use a
variant of \emph{random hyperplane} rounding, pioneered by Goemans and
Williamson \cite{goemans95improved}.  For \maxcut{}, this works
as follows: a random hyperplane passing through the origin is chosen.
This hyperplane naturally induces a cut of the graph: each side
of the cut is defined by the vertices whose vector lies on one side of
the hyperplane. Analyzing the resulting cut boils down to a simple local
argument: one can show that each edge of the graph goes across the cut with probability
at least $\alphagw$ times its contribution to the objective value of the
vectors.

It is helpful to see why the same rounding does not work for \maxbisection{}, i.e.,
why the resulting partition is not necessarily a bisection.
Although each vertex has probability $1/2$ of landing on each side of the cut,
these probabilistic events (for different vertices) are \emph{not} independent.  In fact for some vector
solutions they are highly correlated. In other words although the expected size
of each side of the cut is $|V|/2$, the cut may in general be very unbalanced with high probability.

Most of the previous algorithms have coped with this by coming up with
more sophisticated variants of the random hyperplane rounding that
do produce a partition that is (close to) a bisection.  On the other
hand, the most recent work~\cite{raghavendra12approximating}
took a somewhat different approach.
They use a family of stronger \iSDP\ relaxations derived by the so-called
\la\ lift-and-project system, whose vector solutions enjoy nice
structural properties and which can be rounded to yield an improved approximation
ratio.  As this is not the main contribution of our work, we
only briefly comment on the \la\ lift-and-project system and how it
derives the \iSDP\ that we utilize, in Section~\ref{sec:sdp-prelims}.

The key idea of \cite{raghavendra12approximating} is that using an
operation known as \emph{conditioning},
the \la\ lift-and-project system allows us to obtain solutions to the
standard \iSDP\ in which a typical pair of vertices has very low
correlation.  Therefore, it essentially follows by Chebyshev's
inequality that the size of each side of the partition produced by
hyperplane rounding will be concentrated around $|V|/2$.
Once such a
nearly-balanced partition is found it can be adjusted to a
bisection for a small additive loss in the number of
edges cut.

There is, however, a major caveat hiding in the word ``correlation''
in the paragraph above.  
There are many possible ways of defining what it means for the vectors
to have ``low correlation'', and the precise notion used in the algorithm of \cite{raghavendra12approximating}
results in rather severe constraints on the rounding algorithm that
can be applied to the vectors.  In particular plain vanilla random
hyperplane rounding still does not produce a cut that is close to a
bisection; if it did, we would already have an $(\alpha_{GW}-\epsilon)$-algorithm!

In their $\RTratio$-algorithm, \cite{raghavendra12approximating} used
\emph{thresholded} random hyperplane rounding in the space orthogonal
to $\bv_0$.  In this rounding, each vertex $i$ has a threshold $t_i$ which adjusts the probability that vertex $i$ falls on a given side of the cut (by shifting the hyperplane by $t_i$ along its normal when looking at which side of the hyperplane $\bv_i$ lies).
How one chooses these thresholds $t_i$ is the key to both the balance
and the objective value of the resulting cut. Using a certain natural
choice of thresholds, \cite{raghavendra12approximating} show that the
resulting cut is near-balanced while at the same time providing a
good approximation ratio. The main issue that restricts their method
is that their proof that the resulting cut is near-balanced is only
applicable to their particular choice of the thresholds.

The source of our improved approximation ratio is as follows.  First,
we use a stronger notion of what it means for an \iSDP\ solution
to have ``low correlation'', and show that after minor modifications
the techniques of \cite{raghavendra12approximating} can be extended to
produce \iSDP{} solutions that have low correlation under this
stronger definition.  Then, the advantage of this modification is that it
buys us a lot of freedom to choose the thresholds for the
random hyperplane rounding (though plain random hyperplane rounding is still not possible).
This lets us propose a rich family of algorithms all of which would result in
an near-balanced cut.

As it turns out, the family of roundings algorithms is still quite
restrictive.  While a simple modification to the choice of thresholds from
\cite{raghavendra12approximating} gives an improved ratio of
$\alphalinear$, improving this to $\alphabest$ is more
challenging.  As opposed to all previous similar rounding algorithms
that we are aware of, our procedure for choosing thresholds has a
combinatorial flavor.  This results in an interesting side effect that
we think is worth mentioning: two vertices $i$ and $j$ whose vectors
$\bv_i$ and $\bv_j$ are equal, may be treated completely differently
by the rounding algorithm, i.e., they may have completely
different probabilities of landing on each side of the cut. 
We are not aware of any previous rounding algorithms where
this occurs.

The extra flexibility that comes from this combinatorial component makes the approximation ratio harder to analyze.
In previous algorithms, the probability that an edge $ij$ is cut
only depends on the pairwise inner products between the three vectors $\bv_0, \bv_i, \bv_j$.  Thus computing
the approximation ratio boils down to minimizing a certain function
in three variables. In our algorithm however, the rounding thresholds $t_i$ and $t_j$ of
the vertices $i$ and $j$ -- and hence the probability that the edge is cut --
are not determined by these three vectors.

However, we are able to analytically remove this uncertainty and
reduce the problem of computing the approximation ratio to again
minimizing a certain function in the three inner products.
Unfortunately it is not possible to compute this minimum analytically,
and we resort to a computer assisted proof.  In particular, using a
computer program we can break the space of all possible values for the
inner products of $v_i, v_j, v_k$ into small cubes and then lowerbound
the approximation ratio of the algorithm for each such cube. The
approach is in the same spirit as those in
\cite{zwick02computer,sjogren09rigorous} and produces a rigorous
(albeit very large) proof of Theorem~\ref{thm:main}. The details of
the computer assisted proof are presented in
Section~\ref{sec:num-proofs}.


Our results for \balmaxtwosat{} are easier: the best algorithm
for \maxtwosat{} is already based on a thresholded random hyperplane
rounding and, luckily for us, chooses thresholds in such a way that
the resulting assignment is expected to be close to balanced. In other words
the optimal rounding for \maxtwosat{} is in our family of rounding algorithms
and can be used for \balmaxtwosat{}.

\subsection{Organization}

The rest of the paper is organized as follows.
Section~\ref{sec:prelims} contains some preliminaries and sets up some
notation.  In Section~\ref{sec:general} we describe a fairly general
family of \maxbisection{} algorithms. In fact the algorithm of
\cite{raghavendra12approximating} is the simplest possible algorithm in
our family. We then present a relatively simple improvement over
\cite{raghavendra12approximating} in Section~\ref{sec:first_improvement}.
Then we give our best algorithm in Section~\ref{sec:pairing}, resulting
in our final bound of $\alphabest$.  In Section~\ref{sec:2sat}
we note that the algorithm of Section~\ref{sec:general} can be applied
to \balcsp{P} problems in general, and in particular to
\balmaxtwosat{} for which we immediately obtain Theorem~\ref{thm:2sat}. 
We elaborate further on the details of our computer generated proof in Section~\ref{sec:num-proofs}.
We conclude with some remarks in Section~\ref{sec:conclusions}.


\section{Preliminaries}\label{sec:prelims}

For notational convenience we work with unweighted graphs throughout
the paper, but we note that our algorithm and its analysis applies
verbatim to the weighted case as well.
Given a graph $G=(V,E)$ the \maxbisection{} problem can be formulated as an
integer program as follows.  To each vertex $i \in V$ we associate a
variable $x_i \in \{-1,1\}$, with the two values representing the two
different pieces of the bisection.  The $0$-$1$ indicator of whether
an edge $ij \in E$ is cut can then be written as $\frac{1-x_ix_j}{2}$.  We define
\begin{align*}
  \Val(x) &= \frac{1}{2} \sum_{ij \in E} (1-x_ix_j) \in [0,1]
\end{align*}
to be the number of edges cut by a partition $x \in
\{-1,1\}^n$.  The \maxbisection{} problem is then
\begin{equation}
  \label{eq:max bisection IP}
  \begin{aligned}
    \max \quad & \Val(x) \\
    \textup{s.t.} \quad & \sum_{i \in V} x_i = 0 \\
    & x_i \in \{-1,1\} & \forall   i \in V.
  \end{aligned}
\end{equation}
We denote by $\Opt(G)$ the optimum of the above program, i.e., the number of
edges cut by the optimal bisection.

\subsection{Semidefinite Relaxation and the Lasserre System}
\label{sec:sdp-prelims}

By replacing the $x_i$'s with high dimensional unit vectors $\bv_i$ and their
products by the corresponding inner products, we
obtain the basic \iSDP{} relaxation for \maxbisection{}. 
For a set of unit
vectors $\bv_1, \ldots, \bv_n$, we write
\[\SDPVal(\{\bv_i\}) = \frac{1}{2} \sum_{ij \in E} \left(1-\iprod{\bv_i}{\bv_j}\right)\]
for the objective function of the vectors.  The basic
SDP relaxation is then
\begin{align}
  \max \quad & \SDPVal(\{\bv_i\}) \notag \\
  \textup{s.t.} \quad & \norm[2]{\sum_{i \in V} \bv_i}^2 = \iprod{\sum_{i \in V} \bv_i}{\sum_{i \in V} \bv_i} = 0 \notag\\
  & \iprod{\bv_i}{\bv_i} = 1 & \forall   i \in V. \notag
\end{align}
To strengthen the standard \iSDP{} for \maxbisection{} one can add variables $\bv_S$ for
any small set $S \subset V$ ($|S| \leq \ell$). This variable will
simulate $\prod_{i\in S} x_i$, i.e., the parity of the number of
vertices $i\in S$ on one side of the cut. If one adds a few intuitive
consistency requirements on these variables one gets an
\iSDP{} relaxation which is equivalent to the so-called level-$\ell$
\la{} strengthening of the standard \iSDP{}.
\begin{align}\label{equa: LA mBisection}
  \max \quad & \SDPVal(\{\bv_i\}) \notag \\
  \textup{s.t.} \quad & \iprod{\bv_{\emptyset}}{\sum_{i\in V}\bv_{S \symd \set{i}}} =0 & \forall S \subseteq V, |S| < \ell \notag\\
  & \iprod{\bv_{S_1}}{\bv_{S_2}} = \iprod{\bv_{S_3}}{\bv_{S_4}}& \forall S_1, \ldots, S_4 \subseteq V,\; |S_1|, \ldots, |S_4| \leq \ell,\; S_1 \symd S_2 = S_3 \symd S_4\notag\\
  & \iprod{\bv_\emptyset}{\bv_\emptyset} = 1 \notag
\end{align}
We write $\SDP_{\ell}(G)$ for the optimum of this semidefinite
program. It is not hard to check that this this is valid relaxation for \maxbisection{}, i.e., for all $\ell$, $\SDP_{\ell}(G) \geq \Opt(G)$.

The parameter $\ell$ for us will be a fixed constant that we will choose later. Note that the above program can
be solved in time $n^{O(\ell)}\in \poly(n)$ using semidefinite programming.
We will use $\bv_i$ as a shorthand for $\bv_{\set{i}}$ and $\bv_0$ as a shorthand
for $\bv_\emptyset$.

We note that the above program enjoys many nice properties including a
probabilistic interpretation involving the so-called ``local distributions'',
however as these are not the main focus of the current work we refer
the interested reader to \cite{lasserre02explicit} and \cite{chlamtac11convex}.
We do use the follwing property of the program however.
The vectors $\bv_i$ satisfy the so-called triangle inequalities.
In particular, if $\ell \ge 2$ for any
three vectors $\bu_0 = \pm \bv_0$, $\bu_1, \bu_2 \in \{\pm \bv_1, \ldots,
\pm \bv_n\}$ the following inequality holds:
\begin{equation}
  \label{eq:l-2-2 triangle ineq}
  \norm[2]{\bu_1 - \bu_2}^2 \leq \norm[2]{\bu_1 - \bu_0}^2 + \norm[2]{\bu_0 - \bu_2}^2.
\end{equation}
When analyzing the algorithm, the relevant quantities turn out to be
the pairwise inner products $\iprod{\bv_i}{\bv_0}$,
$\iprod{\bv_j}{\bv_0}$, and $\iprod{\bv_i}{\bv_j}$.
For this reason, we introduce shorthand
notation $\mu_i := \iprod{\bv_i}{\bv_0}$ and $\rho_{ij} := \iprod{\bv_i}{\bv_j}$.
As the $\bv$'s are unit vectors the inequalities \eqref{eq:l-2-2 triangle ineq} are
equivalent to the following inequalities for every $i, j \in [n]$
\begin{equation}
  \label{eq:triangle ineq}
  \begin{aligned}
    \mu_i + \mu_j + \rho_{ij} & \ge -1 &\;\;\; \mu_i - \mu_j - \rho_{ij} & \ge -1\\
    -\mu_i + \mu_j - \rho_{ij} & \ge -1 &\;\;\; -\mu_i - \mu_j + \rho_{ij} & \ge -1.
  \end{aligned}
\end{equation}

This motivates the following definition.

\begin{definition}[Configuration]
  \label{def:conf}
  We denote by $\Conf \subseteq [-1,1]^3$ the polytope defined by
  \eqref{eq:triangle ineq} together with the constraints $\mu_i,
  \mu_j, \rho_{ij} \in [-1,1]$.  A tuple $(\mu_1, \mu_2, \rho) \in
  \Conf$ is called a \emph{configuration}.
\end{definition}

Typical rounding schemes round the
vectors $\bv_i$ by considering their projections on a random vector.
However, while this produces a cut that is balanced in expectation,
it might not be close to balanced with high probability as vertices
might be correlated. 
One of the main ideas in \cite{raghavendra12approximating} is the notion of vectors with low global correlation.  There are many
possibilities for such a notion; \cite{raghavendra12approximating} introduce a notion called $\alpha$-independence.  For our algorithm, we need the following stronger definition.
\begin{definition}[$\eps$-uncorrelated \iSDP{} solution]
  \label{def: eps uncorrelated sdp solutions}
  Let $\bv_0, \ldots, \bv_n$ be a vector solution.  Write $\bw_i =
  \bv_i - \iprod{\bv_0}{\bv_i}\bv_0$ for the part of $\bv_i$ that is
  orthogonal to $\bv_0$, and $\obw_i = \bw_i/\|\bw_i\|$.
  Then, $\bv_0, \ldots, \bv_n$ is \emph{$\epsilon$-uncorrelated} if 
  \[\Exp[i,j \in V]{|\iprod{\obw_i}{\obw_j}|} \leq \epsilon.\]
\end{definition}
For the interested reader that is familiar with the probabilistic
interpretations of the \la{} hierarchy the quantity
$\iprod{\obw_i}{\obw_j}$ precisely equals the \emph{correlation
  coefficient} between the variables $x_i$ and $x_j$.  
In comparison, $\alpha$-independence used in
\cite{raghavendra12approximating} is defined in terms of the
\emph{mutual information} of the same variables, which is within a
quadratic factor of their \emph{covariance}, $\iprod{\bw_i}{\bw_j}$.

\subsection{Normal Distributions}
\label{sec:gaussian prelims}

Throughout the paper, we write $\phi(x) =
\frac{1}{\sqrt{2\pi}}e^{-x^2/2}$ for the density function of a
standard normal random variable, $\Phi(x) = \int_{y = -\infty}^{x}
\phi(y) \ud y$ for its CDF, and $\Phi^{-1}: [0,1] \rightarrow
    [-\infty, \infty]$ for the inverse of $\Phi$.
We also make use of the following standard fact about projections of
Gaussians onto vectors.

\begin{fact}\label{fact: projection on gaussian vectors}
  Let $\bu_1, \ldots, \bu_t \in \R^n$, $\boldg$ an $n$-dimensional standard Gaussian
  vector, and $z_i = \iprod{\bu_i}{\boldg}$.  Then $z_1,
  \ldots, z_t$ are jointly Gaussian random variables with expectation
  $0$ and covariances $\Cov[z_i, z_j] = \iprod{\bu_i}{\bu_j}$.
\end{fact}

We also need notation for the CDF of the bivariate normal distribution.

\begin{definition}\label{def:Gamma}
  Let $\trho \in [-1,1]$.  We define $\Gamma_{\trho}: [0,1]^2 \rightarrow
  [0,1]$ by
  \[\Gamma_{\trho}(x, y) = \Pr\left[X \le \Phi^{-1}(x) \wedge Y \le \Phi^{-1}(y)\right],\]
  where $X$ and $Y$ are jointly normal random variables with mean
  $0$ and covariance matrix $\left(\begin{array}{cc} 1 & \trho \\ \trho
  & 1
  \end{array}\right)$.
\end{definition}

The following parametrization of the $\Gamma$ function is convenient when analyzing our algorithms. 
The motivation will become clear in Section~\ref{sec:analysis of ratio}.

\begin{definition}
  \label{def:Lambda}
  For $\trho \in [-1,1]$, recall the definition of $\Gamma_{\trho}:
  [0,1]^2 \rightarrow [-1,1]$ and define
  $\Lambda_{\trho}: [-1,1]^2 \rightarrow [-1,1]$ as
  \[
  \Lambda_{\trho}(r_1, r_2) = 2 \Gamma_{\trho}\left(\frac{1-r_1}{2}, \frac{1-r_2}{2}\right) + \frac{r_1+r_2}{2}.
  \]
\end{definition}

We now state three lemmata about $\Gamma_{\trho}$ that turn out to be
useful for us.  For completeness, proofs can be found in
Appendix~\ref{sec:gaussian proofs}.

\begin{lemma}
  \label{lemma:gammasymmetry}
  For every $\trho \in [-1,1]$, $q_1, q_2 \in [0,1]$, we have
  \[\Gamma_{\trho}(1-q_1, 1-q_2) = \Gamma_{\trho}(q_1, q_2) + 1 - q_1 - q_2.\]
\end{lemma}

\begin{lemma}
  \label{lemma:gamma indep bound}
  For every $\trho \in [-1,1]$, $q_1, q_2 \in [0,1]$, we have
  \[\Gamma_{\trho}(q_1, q_2) \le q_1 q_2 + 2|\trho|.\]
\end{lemma}

\begin{lemma}
  \label{lemma:gamma deriv}
  For every $\trho \in (-1,1), q_1, q_2 \in [0,1]$, we have
  \begin{align*}
    \frac{\deriv}{\deriv q_1} \Gamma_{\trho}(q_1, q_2) &= \Phi\left(\frac{t_2 - \trho t_1}{\sqrt{1-\trho^2}}\right)
  \end{align*}
  where $t_i = \Phi^{-1}(q_i)$.
\end{lemma}


\section{A Family Of Bisection Algorithms}\label{sec:general}

In this section we describe a general family of rounding algorithms
for \maxbisection{}.  We first describe the following lemma that we need for our algorithm. 

\begin{lemma}
  \label{lem:cor}
  There is an algorithm which, given an integer $t$ and a graph $G =
  (V,E)$, runs in time $n^{O(t)}$ and outputs a set of unit vectors $\bv_0,
  \ldots, \bv_n$ such that
  \begin{enumerate}
  \item $\SDPVal(\{\bv_i\}) \ge \Opt(G) - 10 t^{-1/12}$,
  \item \label{lem:cor:balance} $\sum_{i} \iprod{\bv_0}{\bv_i} = 0$,
  \item The triangle inequalities \eqref{eq:triangle ineq} are satisfied,
  \item \label{lem:cor:uncorr} The vectors $\bv_0, \bv_1,\ldots, \bv_n$ are $t^{-1/4}$-uncorrelated.
  \end{enumerate}
\end{lemma}

Lemma~\ref{lem:cor}, which we prove in Section~\ref{sec:lem cor proof}
below, is analogous to Theorem~\RTAlphaIndepThm{} in the full version of
\cite{raghavendra12approximating}.  The main difference is in item
\ref{lem:cor:uncorr}.  As mentioned in Section~\ref{sec:sdp-prelims},
\cite{raghavendra12approximating} uses the notion of $\alpha$-independence
which bounds the average mutual information in an average pair of variables
$i, j$, corresponding (up to a quadratic factor) to bounding the average
covariance between a pair of variables, whereas the notion of
$\eps$-uncorrelation (Definition~\ref{def: eps uncorrelated sdp solutions})
bounds the average \emph{correlation coefficient}.  We need this stronger
property of the vectors because we use a more general family of rounding
functions than \cite{raghavendra12approximating}.  

The \maxbisection{} algorithm is presented in
Algorithm~\ref{alg:generic}.  It uses a random hyperplane rounding
that is parameterized by a second algorithm, which we refer to as a
\emph{bias selection algorithm}.

\begin{algorithm}
\caption{\maxbisection{} algorithm}
\label{alg:generic}
\begin{algorithmic}[1]
\REQUIRE Graph $G = (V, E)$, parameter $\epsilon > 0$, bias selection algorithm $\textsc{SelectBias}$
\ENSURE Assignment $y \in \{-1,1\}^n$ satisfying $\sum y_i = 0$ 
\STATE Run the procedure of
Lemma~\ref{lem:cor} with $t = (20/\epsilon)^{12}$ to get vectors $\bv_0, \ldots, \bv_n$
\STATE $\mu_i \leftarrow \iprod{\bv_0}{\bv_i}$, $\bw_i \leftarrow \bv_i - \mu_i \bv_0$,\\
$\obw_i \leftarrow \begin{cases}\frac{\bw_i}{||\bw_i||_2}&\mbox{if } ||\bw_i||_2 \neq 0,\\\mbox{a unit vector orthogonal to all other vectors}&\mbox{if } ||\bw_i||_2 = 0\end{cases}$
\STATE $(r_1, \ldots, r_n) \leftarrow \textsc{SelectBias}(\mu_1, \ldots, \mu_n)$ \label{step:bias_select}
\STATE $\boldg \leftarrow$ standard $n$-dimensional Gaussian vector
\STATE $
  x_i \leftarrow \left\{\begin{array}{rl}
  -1 & \text{if $\iprod{\obw_i}{\boldg} < \Phi^{-1}(\frac{1 - r_i}{2})$} \\
  1 & \text{otherwise}
  \end{array}\right.
  $\label{step:x_i}\\
\STATE $b \leftarrow \frac{1}{2} \sum_{i \in V} x_i$ (the imbalance of $x$)
\STATE $S \leftarrow$ a uniformly random set of $|b|$ 
  vertices $i$ s.t.~$x_i = \sign(b)$
\STATE $y_i \leftarrow \left\{\begin{array}{rl}
   x_i & \text{if $i \not\in S$}\\
  -x_i & \text{if $i \in S$}
  \end{array}\right.$\\
\RETURN $y_1, \ldots, y_n$
\end{algorithmic}
\end{algorithm}

To understand the bias selection algorithm,  first note that by
Fact~\ref{fact: projection on gaussian vectors} the
value $\iprod{\obw_i}{\boldg}$, used to determine the value of $x_i$ in step~\ref{step:x_i}, is a
standard Gaussian random variable. It then
follows that $\Exp{x_i} = r_i$, i.e., $r_i$ is precisely
the bias of $x_i$ produced by the rounding algorithm.

Thus, in order for the intermediate cut $x$ to be balanced in
expectation, we require that the output of the bias selection
algorithm satisfies $\sum r_i = 0$.  This could be relaxed to only
requiring that $|\sum r_i| \le \epsilon n$, but we do not need this
relaxed notion.  The bias selection algorithm can be randomized, in
which case, we would require that $\pro{\sum r_i = 0} \approx 1$.  In
principle, the bias selection algorithm is allowed to look at 
the SDP solution $\bv_0, \ldots, \bv_n$ as well as $G$, but our bias
selection algorithm only uses $\mu_1, \ldots, \mu_n$. Notice that
item~\ref{lem:cor:balance} of Lemma~\ref{lem:cor} implies that
$\sum_i \mu_i = 0$.

Varying the bias selection algorithm gives rise to different rounding
algorithms, and the question is how to efficiently find $r_i$'s that
give a good approximation ratio.  The Raghavendra-Tan Algorithm,
achieving an approximation ratio of $\RTratio$, can be expressed in
this framework as choosing $r_i = \mu_i$.  In general, it would be
natural to let $r_i$ depend solely on $\mu_i$, i.e., $r_i := f(\mu_i)$
for some function $f: [-1,1] \rightarrow [-1,1]$.  However, because of
the balance requirement $\sum r_i = 0$, the function $f$ would need to
be linear, resulting in a quite limited family of roundings.
Nevertheless, as we shall see in Section~\ref{sec:first_improvement},
this kind of rounding is sufficient to obtain a
$\alphalinear$-approximation.

In order to improve upon this, the choice of $r_i$ needs to look at
more than just $\mu_i$.  In Section~\ref{sec:pairing}, we devise a
$\alphasecond$-algorithm.  There the bias selection algorithm starts by
setting $r_i = c \cdot \mu_i$ but then adjusts some of the $r_i$
values in a controlled way so as to preserve $\sum r_i = 0$.  Somewhat
curiously, an effect of our rounding scheme is that two vertices $i$
and $j$ of the graph with the same vectors $\bv_i = \bv_j$ can be
rounded differently by the algorithm.  We are not aware of previous
algorithms where this happens.

\subsection{Overview of Analysis}
\label{sec:analysis overview}

To analyze the algorithm, first notice that from Lemma~\ref{lem:cor}
$\SDPVal(\{\bv_i\}) \ge \Opt(G) - \epsilon/2$. Thus, it suffices to
lower bound the value of the resulting bisection $y$ in terms of $\SDPVal(\{\bv_i\})$.

Now consider the intermediate cut $x$ of Algorithm~\ref{alg:generic}.
When constructing $x$, the behaviour of
the algorithm on an edge $(i, j) \in E$ depends solely on the pairwise
inner products $\mu_i$, $\mu_j$, $\rho_{ij}$ and the two
biases $r_i$ and $r_j$.  Notice that by Lemma~\ref{lem:cor}
$(\mu_i$, $\mu_j$, $\rho_{ij})$ is a configuration as defined in
Definition~\ref{def:conf}.  We would then like to compute the
``relative contribution'' $\alpha: \Conf \times [-1,1]^2 \rightarrow
\R_{\ge 0}$ defined such that $\alpha(\mu_i, \mu_j, \rho_{ij}, r_i, r_j)$ is
the contribution of the edge $(i,j)$ to the value of the rounded solution
divided by its contribution to the value of the \iSDP{} solution
$\bv_0, \ldots, \bv_n$.  In other words we define $\alpha$, somewhat informally,
as
\[
\alpha(\mu_i, \mu_j, \rho_{ij}, r_i, r_j) = \frac{\pro{x_i \ne x_j\,|\,\mu_i,\mu_j, \rho_{ij}, r_i, r_j}}{(1-\rho_{ij})/2}.
\]
A formal definition appears in Section~\ref{sec:analysis of ratio}, Definition~\ref{def:alpha}.
Given this definition, the following lemma, which lower bounds the value of
the cut $x$, is intuitively obvious. We prove it in Section~\ref{sec:analysis of ratio}.
\begin{lemma}
  \label{lemma:approx}
  Suppose that for every edge $(i, j) \in E$, it holds that $\alpha(\mu_i,
  \mu_j, \rho_{ij}, r_i, r_j) \ge \alpha$, for some $\alpha \ge 0$.
  Then the assignment $x$ produced by Algorithm~\ref{alg:generic}
  satisfies
  \[\Exp{\Val(x)} \ge \alpha \SDPVal(\{\bv_i\}).\]
\end{lemma}

Finally, we need to show that the balancing step at the end of the
algorithm only incurs a small loss.  The following lemma, proved in
Section~\ref{sec:balance}, establishes this. The main idea is to
show that most of the time the solution $x$ is not too unbalanced to begin with. 
\begin{lemma}
  \label{lemma:balance_loss}
  Consider Algorithm~\ref{alg:generic} and suppose the biases selected
  in step~\ref{step:bias_select} satisfy $\sum r_i = 0$.  Then, it holds that
  \[
  \Exp{\Val(y)} \ge \Exp{\Val(x)} - \epsilon/2.
  \]
\end{lemma}

Taken together, Lemmata~\ref{lem:cor}, \ref{lemma:approx} and
\ref{lemma:balance_loss} imply that Algorithm~\ref{alg:generic} is an
$(\alpha-\epsilon)$-approximation algorithm for \maxbisection{} so the main crux is understanding
the function $\alpha: \Conf \times [-1,1]^2 \rightarrow \R_{\ge 0}$.
We note that the running time of the algorithm is $n^{O(1/\epsilon^{12})}$.

\subsection{Analysis of Approximation Ratio}
\label{sec:analysis of ratio}

In this section we elaborate further on the definition of the function
$\alpha$, and Lemma~\ref{lemma:approx}.  First we express the probability
that two vertices are on the same side of the cut. Recall Definition~\ref{def:Lambda}
of the function  $\Lambda_{\trho}$.

\begin{lemma}
  \label{lemma:Lambda}
  Consider Algorithm~\ref{alg:generic} and any pair of variables $x_i,
  x_j$.  Denote $\trho = \iprod{\obw_i}{\obw_j}$.  Then,
  \[\pro{x_i = x_j} = \Lambda_{\trho}(r_i, r_j).\]
\end{lemma}

\begin{proof}
  By Fact~\ref{fact: projection on gaussian vectors},
  $\iprod{\obw_i}{\boldg}$ and $\iprod{\obw_j}{\boldg}$ are jointly normal with
  covariance $\trho$ and variance $1$.  Thus,
  \begin{align*}
    \pro{x_i = -1 \wedge x_j = -1} &= \Gamma_{\trho}\left(\frac{1-r_i}{2}, \frac{1-r_j}{2}\right)\\
    \pro{x_i = 1 \wedge x_j = 1} &= \pro{\iprod{\obw_i}{\boldg} \ge \Phi^{-1}\left(\frac{1-r_i}{2}\right) \wedge \iprod{\obw_j}{\boldg} \ge \Phi^{-1}\left(\frac{1-r_j}{2}\right)}\\
    &= \pro{\iprod{\obw_i}{\boldg} < -\Phi^{-1}\left(\frac{1-r_i}{2}\right) \wedge \iprod{\obw_j}{\boldg} < -\Phi^{-1}\left(\frac{1-r_j}{2}\right)} \\
    &=\Gamma_{\trho}\left(\frac{1+r_i}{2}, \frac{1+r_j}{2}\right),
  \end{align*}
  where the middle step used that $\boldg$ and $-\boldg$ have the same distribution and the last step used
  $\Phi(-x) = 1-\Phi(x)$.  Using Lemma~\ref{lemma:gammasymmetry} we get
  \[
  \pro{x_i = x_j} = \pro{x_i = -1 \wedge x_j = -1} + \pro{x_i = 1 \wedge x_j = 1} = 2 \Gamma_{\trho}\left(\frac{1-r_i}{2}, \frac{1-r_j}{2}\right) + \frac{r_i}{2} + \frac{r_j}{2}.\qedhere
  \]
\end{proof}

We are now ready to give a definition of the function $\alpha$ described above.
\begin{definition}
  \label{def:alpha}
  For a configuration $(\mu_1, \mu_2, \rho) \in \Conf$, and for $r_1, r_2 \in [-1,1]$, let $\trho
  = \frac{\rho-\mu_1\mu_2}{\sqrt{(1-\mu_1^2)(1-\mu_2^2)}}$ (if $\mu_1 = \pm 1$
  or $\mu_2 = \pm 1$ we let $\trho = 0$), and define
  \[
  \alpha(\mu_1, \mu_2, \rho, r_1, r_2) = \frac{2\left(1 - \Lambda_{\trho}(r_1, r_2) \right)}{1 - \rho}.
  \]
\end{definition}

From the discussion above, we can immediately deduce Lemma~\ref{lemma:approx}.

\begin{proof}[Proof of Lemma~\ref{lemma:approx}]
When we run Algorithm~\ref{alg:generic}, the probability we cut an edge $ij$ is
\[ \pro{x_i \not = x_j}  = 1 -  \pro{x_i = x_j} = \frac{\left(1 - \rho_{ij}\right)\alpha(\mu_i,\mu_j, \rho_{ij}, r_i, r_j)}{2} \geq \frac{\left(1 - \rho_{ij}\right)\alpha}{2}.\]
We remind the reader that $\frac{\left(1 - \rho_{ij}\right)}{2}$ is exactly the contribution of the pair $ij$ to $\SDPVal(\{\bv_i\})$, hence,
\[\Exp{\Val(x)} = \sum_{ij \in E} \pro{x_i \not = x_j} \ge \alpha \SDPVal(\{\bv_i\}).\qedhere\]
\end{proof}

\subsection{Analysis of Balance}
\label{sec:balance}

In this section we prove Lemma~\ref{lemma:balance_loss}.  The lemma
follows immediately from the following lemma.  

\begin{lemma}\label{lemma:almost_balanced}
  Consider Algorithm~\ref{alg:generic} and assume that the biases selected
  in step~\ref{step:bias_select} satisfy $\sum r_i = 0$. Then
  the assignment $x$ chosen in step~\ref{step:x_i} satisfies
  \begin{equation*}
    \pro[\boldg]{\left|\sum_i x_i\right| \ge \epsilon n/10} \leq \epsilon/10.
  \end{equation*}
\end{lemma}

This lemma is analogous to (but does not follow from)
Theorem~\RTVarianceThm{} and Corollary~\RTBalanceCor{} in the full version of \cite{raghavendra12approximating}.
The main difference is that our lemma applies to any
choice of biases, whereas Theorem~\RTVarianceThm{} in
\cite{raghavendra12approximating} requires $r_i = \mu_i$.  This is
enabled by our stronger notion of an uncorrelated SDP
solution, i.e., Lemma~\ref{lem:cor}. The proof of Theorem~\RTVarianceThm{} in \cite{raghavendra12approximating} 
is an elegant use of information-theorectic techniques ultimately
relying on the so-called data processing inequality. While one can easily
extend that proof to our setting, we give a different proof which is somewhat
longer, but in our opinion more transparent and resulting in somewhat better bounds
as we do not need to pass back and forth between covariance and mutual
information.

\begin{proof}[Proof of Lemma~\ref{lemma:almost_balanced}]
  Define the random variable $X = \frac{1}{n} \sum_i x_i$. We have
  $\Exp{X} = \frac{1}{n} \sum_i r_i = 0$.  We now bound
  $\Var{X} = \Exp[i,j \in V]{\Cov[x_i, x_j]}$ from which the desired bound follows using 
  Chebyshev's inequality.  Let $\tau = 
  1/t^{1/4}$ and recall that by Lemma~\ref{lem:cor} we have
  \[\Exp[i,j \in V]{|\iprod{\obw_i}{\obw_j}|} \le \tau.\]
  Fix some pair $i, j$ and let $\trho = \iprod{\obw_i}{\obw_j}$.  By
  Lemma~\ref{lemma:Lambda} we have
  \[ \Cov[x_i,x_j] = \Exp{x_ix_j} - \Exp{x_i}\Exp{x_j} = 2\pro{x_i = x_j} - 1 - r_i r_j = 2\Lambda_{\trho}(r_i, r_j) - 1 - r_i r_j = 4 \Gamma_{\trho}\left(q_i, q_j\right) - 4 q_i q_j,\]
  where $q_i = \frac{1-r_i}{2} = \pro{x_i = -1}$.
  By Lemma~\ref{lemma:gamma indep bound}, it follows that $\Cov[x_i,
    x_j] \le 8 |\trho|$.  Averaging over all pairs $i, j$ we get
  \[\Var{X} = \Exp[i,j \in V]{\Cov[x_i, x_j]} \le \Exp[i,j \in V]{8 |\iprod{\obw_i}{\obw_j}|} \le 8\tau.\]
  Denoting by $\sigma = \sqrt{\Var{X}}$ it follows from Chebyshev's inequality that
  \[\pro{|X| \ge \sigma^{2/3}} \le \sigma^{2/3}. \]
  Plugging in our bound $\sigma^2 \le 8 \tau = 8 t^{-1/4}$ we have
  $\sigma^{2/3} \le 2t^{-1/12} = \epsilon/10$, completing the proof.
\end{proof}

\subsection{Finding the Uncorrelated \iSDP\ Solution}\label{sec:lem cor proof}

In this section we prove Lemma~\ref{lem:cor} which is restated below for convenience.
\begin{relemma}{lem:cor}
  There is an algorithm which, given an integer $t$ and a graph $G =
  (V,E)$, runs in time $n^{O(t)}$ and outputs a set of vectors $\bv_0,
  \ldots, \bv_n$ such that
  \begin{enumerate}
  \item $\SDPVal(\{\bv_i\}) \ge \Opt(G) - 10 t^{-1/12}$,
  \item $\sum_{i} \iprod{\bv_0}{\bv_i} = 0$,
  \item The triangle inequalities \eqref{eq:triangle ineq} are satisfied
 \item The vectors $\bv_1,\ldots, \bv_n$ are $t^{-1/4}$-uncorrelated with respect to $\bv_0$.
  \end{enumerate}
\end{relemma}
\begin{proof}
  Without loss of generality assume that $t$ is a sufficiently large constant and construct random
  vectors $\bv'_0, \bv'_1, \ldots, \bv'_n$ by applying Lemma~\RTMutualInfLemma{} from the full
  version of \cite{raghavendra12approximating}. 
  The \emph{expected} objective value of $\bv'_i$'s (where the expectation is over
  the randomness in their Lemma~\RTMutualInfLemma{}) equals $\SDP_{t+2}(G)\geq \Opt(G)$.
  Furthermore, the same lemma guarantees that the average \emph{mutual information} of certain random
  variables (associated with the vectors), indicated by $I(X_i, X_j)$, is low. We will not define
  mutual information and refer the interested reader to \cite{raghavendra12approximating} as we can
  immediately apply Lemma~\RTCovarMlInfLemma{} from that paper which relates $I(X_i, X_J)$ to $|\iprod{\bw'_i}{\bw'_j}|$ to
  arrive at the following conclusion.
  \begin{align*}
    \Exp[\{\bv_i'\}]{\Exp[i,j \in V]{|\iprod{\bw'_i}{\bw'_j}|}} \leq \Exp[\{\bv_i'\}]{\Exp[i,j \in V]{\sqrt{32 I(X_i, X_j)}}}
    \leq \sqrt{32 \Exp[\{\bv_i'\}]{\Exp[i,j \in V]{I(X_i, X_j)}}} \leq \frac{\sqrt{32}}{\sqrt{t-1}} < \frac{6}{\sqrt{t}},
  \end{align*}
  where $\bw'_i$'s are defined from $\bv'$ analogous to how $\bw_i$'s are defined from $\bv_i$'s.
  Furthermore the $\bv'_i$'s are a solution to level-2 \la{} relaxation of \maxbisection{} and in particular satisfy
  $\sum_i \iprod{\bv'_0}{\bv'_i} = 0$ along with all triangle inequalities. Applying two Markov bounds we conclude,
  \begin{align*}
    \pro[\{\bv_i'\}]{\SDPVal(\{\bv'_i\}) \le \Opt(G) - 9 t^{-1/12}} &\le 1 - 9 t^{-1/12},&
    \pro[\{\bv_i'\}]{\Exp[i,j \in V]{|\iprod{\bw'_i}{\bw'_j}|} \geq t^{-5/12}} &\le 6 t^{-1/12}.
  \end{align*}
  Thus, by resampling $\bv'_0, \bv'_1, \ldots, \bv'_n$ an (expected) $3t^{1/12}$ times we obtain an \iSDP\ solution where all
  the following three conditions as well as the triangle inequalities on $\bv'_i$'s hold:
  \begin{align*}
    \SDPVal(\{\bv'_i\}) &\ge \Opt(G) - 9 t^{-1/12},&     \Exp[i,j \in V]{|\iprod{\bw'_i}{\bw'_j}|} &\leq t^{-5/12},
    & \sum_{i \in V} \iprod{\bv'_0}{\bv'_i} &= 0.
  \end{align*}

  The above vectors have all the required conditions of the Lemma
  except the last. In particular, we have a bound on the average of
  the inner products $\iprod{\bw'_i}{\bw'_j}$ as opposed to the
  stronger bound on the inner products $\iprod{\obw_i}{\obw_j}$.
  But the only way in which the stronger bound can fail to hold
  is if many $\|\bw'_i\|$'s are small, i.e., if many of the $|\mu_i|$
  values are close to $1$.  However, such vectors can be corrected for
  a small additional loss in \iSDP\ value.

  In particular, define vectors $\bv_0, \ldots, \bv_n$ as $\bv_0 = \bv'_0$ and
  \begin{align*}
    \bv_i &=
    \begin{cases}
      \bv'_i & \mbox{if } \|\bw'_i\| \ge t^{-1/12} \\
      \bv'_i - \bw'_i + \bw^*_i & \mbox{otherwise}
    \end{cases}
  \end{align*}
  for $i \ge 1$.  Here $\bw^*_i$ is a new vector orthogonal to all
  other vectors and of length $\|\bw^*_i\| = \|\bw'_i\|$.

  Notice that now, 
  \begin{align*}
    |\iprod{\obw_i}{\obw_j}| &= 
    \begin{cases}
      0 &\mbox{ if }\min(\|\bw'_i\|, \|\bw'_j\|) < t^{-1/12}\\ 
      \frac{|\iprod{\bw'_i}{\bw'_j}|}{\|\bw'_i\| \|\bw'_j\|} & \mbox{otherwise,}
    \end{cases}
  \end{align*}
  which in particular is bounded by $|\iprod{\bw'_i}{\bw'_j}|/t^{-2/12}$
  and therefore
  \[\Exp[i,j \in V]{|\iprod{\obw_i}{\obw_j}|} \leq t^{-5/12}/t^{-2/12} = t^{-1/4}.\]
  Furthermore we can bound the difference in any inner product by 
  \[|\iprod{\bv_i}{\bv_j} - \iprod{\bv'_i}{\bv'_j}| = |\iprod{\bw_i}{\bw_j} - \iprod{\bw'_i}{\bw'_j}| \le t^{-1/12},\]
  and so $\SDPVal(\{\bv_i\}) \ge \SDPVal(\{\bv'_i\}) - t^{-1/12}$.
  Clearly, the condition $\sum \iprod{\bv_0}{\bv_i} = 0$ is still
  satisfied since all projections on $\bv_0$ remain the same.  It
  remains to check the triangle inequalities.  Consider any pair
  $\bv_i, \bv_j$ such that one of them was changed.  We then have
  $\rho_{ij} = \iprod{\bv_i}{\bv_j} = \mu_i \mu_j$, and the four
  inequalities \eqref{eq:triangle ineq} are equivalent to
  \[(1 \pm \mu_i) (1 \pm \mu_j) \ge 0\]
  which clearly hold.
\end{proof}


\section{Linear Biases: A \alphalinear{}-Approximation}\label{sec:first_improvement}

In this section we study how far one can get by considering bias
selection algorithms that set $r_i$ to be a linear function of $\mu_i$.
Recall that the Raghavendra-Tan algorithm, which uses $r_i = \mu_i$,
falls into this category.

\begin{definition}
  \label{def:linear_alpha}
  For $c \in [0,1]$, define
  \[\alpha(c) := \min_{(\mu_1, \mu_2, \rho) \in \Conf} \alpha(\mu_1, \mu_2, \rho, c \cdot \mu_1, c \cdot \mu_2).\]
\end{definition}

The following Lemma is an immediate corollary of the analysis in Section~\ref{sec:analysis overview}.

\begin{lemma}
  For any $0 \leq c \leq 1$, Algorithm~\ref{alg:generic} with the bias selection
  algorithm that sets $r_i = c \cdot \mu_i$ has approximation ratio at
  least $\alpha(c) -\epsilon$.
\end{lemma}

\begin{numericalclaim}\label{claim: best alpha from lin biases}
  $\max_{c \in [0,1]} \alpha(c) \ge \alphalinearfull$, and it is achieved for $c \approx \alphalinearbestc$.
\end{numericalclaim}

Figure~\ref{fig:alpha_linear_plots} shows plots of $\alpha(c)$ for $c
\in [0,1]$ and $c \in [0.83, 0.88]$.  

Just like with our $\alphasecond$-algorithm (to be presented in the
next section), we can obtain a rigorous proof of a slightly weaker
version of Claim~\ref{claim: best alpha from lin biases}.  In
particular we prove that for $c = 0.86451$ (a slight modification of)
Algorithm~\ref{alg:generic} gives a $\alphalinear$-approximation.
This is done in Theorem~\ref{thm:rig-linear}.

\begin{figure}
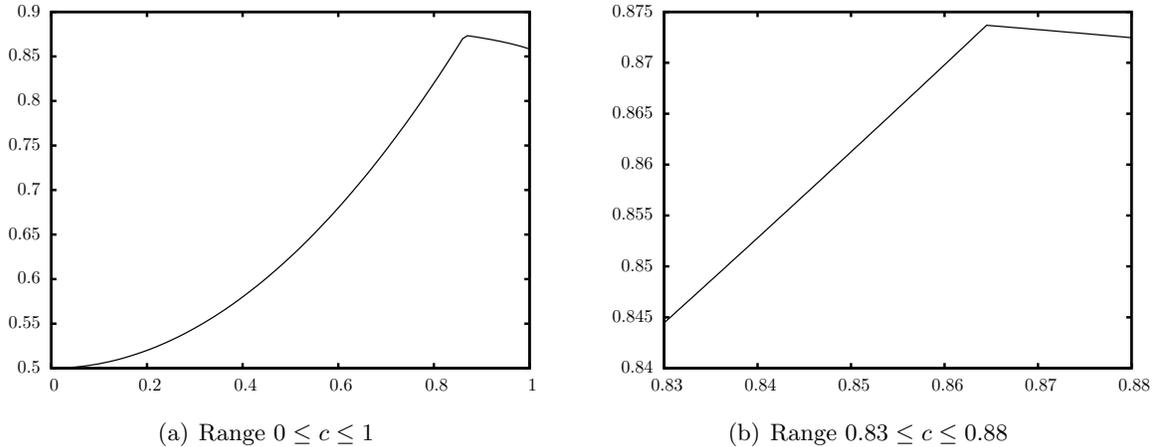

  \centering
  \subfigure[Range $0 \le c \le 1$]{\scalebox{0.62}{\import{plots/}{alpha-linear-full-src.tex}}}
  \subfigure[Range $0.83 \le c \le 0.88$]{\scalebox{0.62}{\import{plots/}{plots/alpha-linear-zoom-src.tex}}}
  \caption{Plot of $\alpha(c)$ for entire range and zoomed in around optimal $c$.}
  \label{fig:alpha_linear_plots}
\end{figure}

Let us now spend some time discussing the worst case configurations for $\alpha(c)$, as our
understanding of these will guide our choices when obtaining further
improvements. Denote by $c^* \approx \alphalinearbestc$ the optimal value of $c$.  
It turns out that (up to symmetry\footnote{Due to symmetry, the configuration
  $(-\mu_1, -\mu_2, \rho)$ is completely equivalent to $(\mu_1, \mu_2, \rho)$.})
there are two distinct worst case configurations $\phi_1, \phi_2$ for $\alpha(c^*)$, approximately
\begin{align*}
  \phi_1 &= (\worstconfiglinearb, \worstconfiglinearb, \worstconfiglinearrho) &
  \phi_2 &= (1, -1, -1).
\end{align*}
The presence of the integral configuration $(1, -1, -1)$ may seem
surprising at first, but has a very natural explanation.  For this
configuration, we have $\trho = 0$, meaning that the two vertices are
rounded completely independently, one with expectation $c$ and the
other with expectation $-c$.  Thus the probability that such an edge
is cut by the algorithm is precisely $\frac{1+c^2}{2}$, and since the
\iSDP\ value for this configuration is $1$ this implies an upper bound
of $\alpha(c) \le \frac{1+c^2}{2}$, meaning that $c$ needs to be sufficiently large in order
for us to obtain a good approximation ratio.  Indeed, the $c < c^*$
part of Figure~\ref{fig:alpha_linear_plots} follows this curve.

The other worst case configuration is more interesting, and is quite
similar to the kind of configuration that is the worst for the
Goemans-Williamson \maxcut{} algorithm.  On this configuration, the
approximation ratio improves as $c$ decreases. Intuitively, this is
because the configuration has both vertices biased in the same
direction, so putting less importance on the bias results in a greater
probability that the edge is cut.  The optimal choice $c^*$ is the
point where the ratio on $\phi_1$ meets the curve $\frac{1+c^2}{2}$.

\subsection{Limitations}

Even though $\max \alpha(c) \approx \alphalinear$, it is possible that
a better ratio could be obtained by choosing $c$ adaptively after
seeing the graph $G$ and \iSDP\ solution $\bv_0, \ldots, \bv_n$. 
To rule out the possibility of any significant improvement of
this form, we exhibit a distribution $D_{\Conf}$ over configurations $(\mu_1,
\mu_2, \rho)$ such that
\begin{equation}
  \label{eqn:mixed_performance}
\max_{c \in [0,1]} \frac{\Exp[(\mu_1, \mu_2, \rho) \sim D_{\Conf}]{1 - \Lambda_{\trho}(c \cdot \mu_1, c \cdot \mu_2)}}{\Exp[(\mu_1, \mu_2, \rho) \sim D_{\Conf}]{(1-\rho)/2}} \le \alphalinearlimit.
\end{equation}
The distribution is quite simple, and is only supported on the two
worst-case configurations for $\alpha(c^*)$.  Specifically, $(\mu_1, \mu_2,
\rho) \sim D_{\Conf}$ is chosen as
\begin{equation*}
  \begin{array}{lllll}
    \phi_1 && \text{with probability \alphalinearlimitp{},}\\
    \phi_2 && \text{otherwise.}
  \end{array}
\end{equation*}

In Figure~\ref{fig:alpha_linear_limit} we plot the approximation ratio
on $\phi_1$ and $\phi_2$ as a function of $c$, as well as the ratio of
\eqref{eqn:mixed_performance} as a function of $c$.  While the
latter curve might appear to be a constant, it does have small
variations of order $10^{-4}$.

\begin{figure}
  \centering
  \scalebox{1}{\import{plots/}{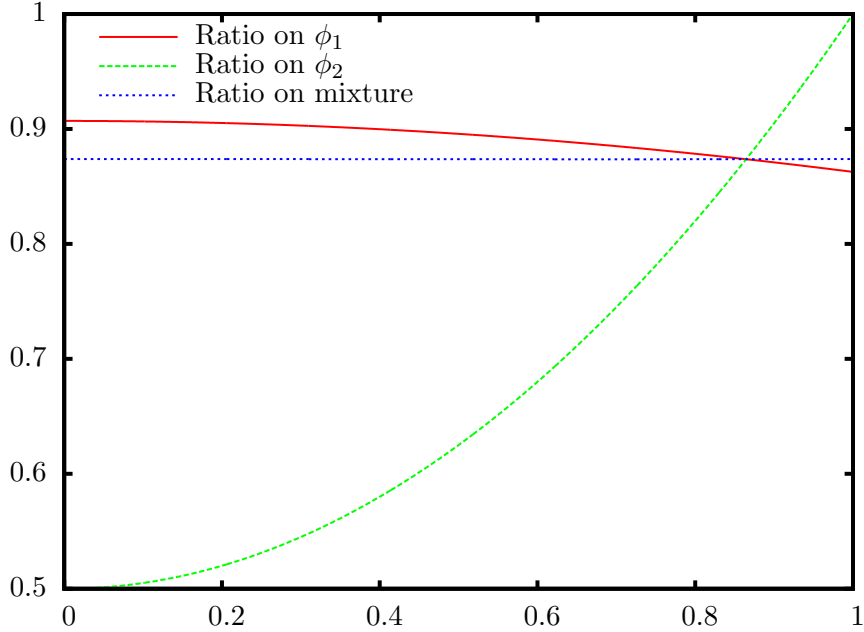}}
  \caption{Approximation ratio on the two configurations and their mixture as a function of $c$}
  \label{fig:alpha_linear_limit}
\end{figure}


\section{Pairing Vertices: A \alphasecond-approximation}\label{sec:pairing}

In this section we describe a bias selection algorithm which yields a
$\alphasecond$-approximation for \maxbisection{}.
Let us start with an informal description of how to obtain the
improvement.  Recall from the discussion on the algorithm in
Section~\ref{sec:first_improvement} that an obstacle to further
improvements was the ``conflict'' between the two critical
configurations $\phi_1$ which resembled a critical configuration for
\maxcut{} and $\phi_2$ which was just an integral configuration.
Arguably, configurations like $\phi_1$ in some sense capture the
difficulty of \maxbisection{}, whereas the integral configuration
$\phi_2$ should be easy.  With this in mind, it is natural to decrease
the value of $c$ to perform better on $\phi_1$ and similar configuration, and then
do some adjustments on vertices with large $|\mu_i|$ in order to perform
well on $\phi_2$ and more generally on near-integral configurations.

A first idea is the following: as long as there are edges $(i,j)$
which are near-integral in the \iSDP\ solution, say, $\mu_i > 1-\delta$ and
$\mu_j < -1+\delta$ for some small constant $\delta$, set $r_i =
1-\delta$ and $r_j = -1+\delta$, respectively.  Once all such edges
are processed, use $r_i = c \cdot \mu_i$ for all other vertices.  Once one
takes care of some technical details this idea can be made to work,
however the improvement over the algorithm of the previous section
is minor, of order, say, $10^{-4}$.

In order to get a more impressive improvement, we use a ``smooth''
version of the above idea.  As in the linear bias selection, we start
by assigning $r_i = c \cdot \mu_i$ for all $i$.  We then pick off pairs of
vertices $(i, j)$ such that $\mu_i > 0$ and $\mu_j < 0$ are as large
as possible (in absolute value).  We then add some value $\Delta r >
0$ to $r_i$ and subtract $\Delta r$ from $r_j$.  Clearly, this
operation preserves $\sum r_i = 0$.  The remaining choice is now how
to choose the ``boost'' $\Delta r$.  It is somewhat natural to
restrict ourselves to choosing $\Delta r := (1-c) f(\min(|\mu_i|, |\mu_j|))$
where $f: [0,1] \rightarrow [0,1]$ is a non-decreasing function which for
technical reasons we require to be Lipshitz continuous and satisfy $f(0) = 0$.
We refer to any such
$f$ as a ``boost function''.  
Notice that before the boosting all biases are in the interval $[-c, c]$
so after the boosting all biases are in $[-1, 1]$, i.e., valid.
Ultimately we choose $f$ to be piece-wise linear though it is quite possible
that further improvements are possible with more complicated choices of $f$.  
More formally, our bias values are given by Algorithm~\ref{alg:biasselect}.

\begin{algorithm}
\caption{\maxbisection{} bias selection}
\label{alg:biasselect}
\begin{algorithmic}[1]
\REQUIRE $\mu_1, \ldots, \mu_n \in [-1,1]$, parameter $c \in [0,1]$, boost function $f: [0,1] \rightarrow [0,1]$
\ENSURE Biases $r_1, \ldots, r_n \in [-1,1]$ such that $\sum r_i = 0$.
\STATE $r_i \leftarrow c \cdot \mu_i$ for $i \in [n]$.
\STATE $S \leftarrow V$
\WHILE{$\max_{i \in S} \mu_i > 0 \wedge \min_{j \in S} \mu_j < 0$}
\STATE $i \leftarrow \argmax_{i \in S} \mu_i$
\STATE $j \leftarrow \argmin_{j \in S} \mu_j$
\STATE $\beta \leftarrow \min(|\mu_i|, |\mu_j|)$
\STATE $r_i \leftarrow r_i + (1-c)f(\beta)$
\STATE $r_j \leftarrow r_j - (1-c)f(\beta)$
\STATE $S \leftarrow S \setminus \{i,j\}$
\ENDWHILE
\RETURN $r_1, \ldots, r_n$
\end{algorithmic}
\end{algorithm}

Let us now analyze the performance of the algorithm.  As opposed to
the linear bias selection algorithm used in the previous section,
given some configuration $(\mu_1, \mu_2, \rho)$ we do not know
exactly what $r$-values were used to round it.  However, we do have
the following Lemma, which provides bounds on these $r$-values.
\begin{lemma}
  \label{lem:rvaluebound}
  For any vertex $i$, the value $r_i$ produced by
  Algorithm~\ref{alg:biasselect} satisfies
  \begin{align}
    \label{eqn:weak_r_bound}
    \sgn(r_i) &= \sgn(\mu_i) &
    c |\mu_i| &\le |r_i| \le c |\mu_i| + (1-c)f(|\mu_i|) \le 1.
  \end{align}
  Furthermore, for any vertex $j$ such that $\sgn(\mu_i) \ne \sgn(\mu_j)$ one of the following two hold,
  \begin{align}
    \label{eqn:strong_r_bound}
    \begin{aligned}
      |r_j| & \ge c |\mu_j| + (1-c)f(\min(|\mu_i|, |\mu_j|)), \text{ or, }  \\
      |r_i| &\ge c |\mu_i| + (1-c)f(\min(|\mu_i|, |\mu_j|)).
    \end{aligned}
  \end{align}
\end{lemma}

In other words, for a pair of vertices whose $\mu$-values are of opposite sign,
at least one of them picks up a ``boost'' which is as large as the
boost of the smaller of the two. 

\begin{proof}[Proof of Lemma~\ref{lem:rvaluebound}]
  The first part, \eqref{eqn:weak_r_bound}, is straightforward: the
  value of $r_i$ is initialized to $c \mu_i$ which clearly satisfies
  \eqref{eqn:weak_r_bound}.  After this, it is changed at most once,
  in which case it has the value $(1-c)f(\beta)$ added or subtracted
  to it depending on the sign of $\mu_i$, and by monotonicity of $f$
  and the fact that $\beta = \min(|\mu_i|, |\mu_{j'}|)$ for some $j'$
  we have $f(\beta) \le f(|\mu_i|)$.

  For the second part, \eqref{eqn:strong_r_bound}, notice that at least
  one of $i$ and $j$ has to be selected in the loop of the algorithm.
  We consider two cases, depending on which was selected first.  Suppose
  $j$ was selected before or in the same iteration
  as $i$.  It was then selected together with some vertex $i' \in V$
  such that $\sgn(\mu_{i'}) = -\sgn(\mu_j) = \sgn(\mu_i)$ and $|\mu_{i'}| \ge |\mu_i|$.
  Thus the boost given to $j$ was at least
  \[(1-c) f(\min(|\mu_j|, |\mu_{i'}|)) \ge (1-c) f(\min(|\mu_j|, |\mu_i|)), \]
  where we have used monotonicity of $f$.

  The other case, when vertex $i$ is selected before vertex $j$, is
  completely symmetric.
\end{proof}

\begin{definition}
  Given $\mu_1, \mu_2 \in [-1,1]$ the \emph{permissible biases}
  $R_{c,f}(\mu_1,\mu_2) \subseteq [-1,1] \times [-1,1]$
  of the pair are all values of $r_1, r_2$ satisfying
  \eqref{eqn:weak_r_bound} if $\sgn(\mu_1) = \sgn(\mu_2)$ and 
  (\ref{eqn:weak_r_bound}-\ref{eqn:strong_r_bound}) if $\sgn(\mu_1) \ne \sgn(\mu_2)$.

  Notice that the permissible biases of a pair $(\mu_1,\mu_2)$ depends on the
  parameters $c \in [0,1]$ and $f: [0,1] \rightarrow [0,1]$, a monotone function
  satisfying $f(0) = 0$, of Algorithm~\ref{alg:biasselect} hence the notation $R_{c, f}$.
\end{definition}

Note that when $\sgn(\mu_1) = \sgn(\mu_2)$, $R_{c,f}(\mu_1,
\mu_2)$ is of the form $I_1 \times I_2$ where $I_1$ (resp. $I_2$) is the interval of
$r$-values with the same sign as $\mu_1$ (resp. $\mu_2$) satisfying \eqref{eqn:weak_r_bound}.
When $\sgn(\mu_1) \ne \sgn(\mu_2)$ then $R_{c,f}(\mu_1,\mu_2)$
can be similarly written as a union $I_1 \times I_2 \cup I_1' \times
I_2'$, corresponding to which of the two variables $r_1$ and $r_2$ is
subject to the stronger bound of \eqref{eqn:strong_r_bound}.

Now, we can lower bound the approximation ratio of the resulting
algorithm by computing the minimum of $\alpha(\mu_1, \mu_2, \rho, r_1,
r_2)$ over all permissible biases. This motivates the following definitions.

\begin{definition}
  \label{def:pairing alpha}
  For $c \in [0,1]$ and a boost function $f: [0,1] \rightarrow [0,1]$, define
  \[\alpha_{c,f}(\mu_1, \mu_2, \rho) = \min_{r_1, r_2 \in R_{c,f}(\mu_1,\mu_2)} \alpha(\mu_1, \mu_2, \rho, r_1, r_2),\]
  and let
  \[\alpha(c, f) = \min_{(\mu_1, \mu_2, \rho) \in \Conf} \alpha_{c,f}(\mu_1, \mu_2, \rho).\]
\end{definition}

An immediate corollary of Lemma~\ref{lem:rvaluebound} and
Lemma~\ref{lemma:approx} is that, for a fixed value of $c$ and $f$,
the approximation ratio when using Algorithm~\ref{alg:biasselect} to
select biases is at least $\alpha(c, f)$. 
Thus, for a given $c$ and $f$ the approximation ratio of the algorithm
can be computed as a five-dimensional optimization problem.  For a
general $f$ however, the domain of feasible points may not even be
convex.  While it turns out that the problem is convex for the $f$
that we ultimately use, we show that the optimization over $r_1$ and
$r_2$ can be eliminated so that we are again left with a minimization
problem over the space of all configurations. This significantly simplifies 
the computations needed to evaluate $\alpha(c, f)$ and makes our computer-assisted case analysis feasible.

As $R_{c,f}(\mu_1,\mu_2)$ is either of the form $I_1 \times I_2$ or $I_1
\times I_2 \cup I_1' \times I_2'$ for some intervals $I_1, I_2, I_1',
I_2' \subseteq [-1,1]$, we make the following definition.
\begin{definition}
  For a configuration $\mu_1, \mu_2, \rho$ and two closed intervals
  $I_1 = [a_1, b_1], I_2 = [a_2, b_2] \subseteq [-1,1]$, define
  \begin{equation}
    \label{eq:worstr}
    \alpha(\mu_1, \mu_2, \rho, I_1, I_2) = \min_{\substack{r_1 \in I_1\\r_2 \in I_2}} \alpha(\mu_1,
    \mu_2, \rho, r_1, r_2)
  \end{equation}
\end{definition}

Minimizing $\alpha$ over $(r_1, r_2) \in R_{c,f}(\mu_1,\mu_2)$ boils
down to at most two computations of $\alpha(\mu_1, \mu_2, \rho, I_1,
I_2)$.  We have the following theorem, which narrows the search over
$r_1$, $r_2$ down to at most nine different possibilities. The proof is rather 
technical and is left for the end of the current section.
\begin{lemma}
  \label{lemma:worstr} 
  For every configuration $\mu_1, \mu_2, \rho$ and closed intervals
  $I_1 = [a_1, b_1], I_2 = [a_2, b_2]$, we have that 
  \[\alpha(\mu_1, \mu_2, \rho, I_1, I_2) = \min_{(r_1, r_2) \in S \cap I_1 \times I_2} \alpha(\mu_1, \mu_2, \rho, r_1, r_2),\]
  where $S$ is defined as follows. Recall that $\trho = \frac{\rho-\mu_1\mu_2}{\sqrt{(1-\mu_1^2)(1-\mu_2^2)}}$.
  \begin{itemize}
  \item If $\trho \le 0$ then $S$ is the extreme points of the set $I_1 \times I_2$, i.e.
    \[S = \set{(a_1, a_2), (a_1, b_2), (b_1, a_2), (b_1, b_2)}.\]
  \item If $\trho > 0$ then $S$ is the extreme points plus five extra points defined in terms of the function
    $g(x) = 1 - 2 \Phi\left(\Phi^{-1}\left(\frac{1-x}{2}\right) / \trho\right)$. More precisely,
    \[S = \{(0, 0), (a_1, a_2), (a_1, b_2), (b_1, a_2), (b_1, b_2), (a_1, g(a_1)), (b_1, g(b_1)), (g(a_2), a_2), (g(b_2), b_2)\}.\]
  \end{itemize}
\end{lemma}

In other words, if the optimizer for \eqref{eq:worstr} is not $(0,0)$, one of
the $r$-values is always at an extreme point of its domain, and the
other is either at an extreme point, or at a point directly
computable from the first value.  In fact, since the permissible
intervals $R_{c,f}(\mu_1, \mu_2)$ only intersect the origin $(0,0)$ at
their endpoints, the possibility $(0,0)$ can be discarded when
computing $\alpha_{c,f}(\mu_1,\mu_2)$.  Thus, the value of
$\alpha_{c,f}(\mu_1, \mu_2, \rho)$ can be computed by evaluating
$\alpha(\mu_1, \mu_2, \rho, r_1, r_2)$ on at most $16$ possible bias
pairs $(r_1, r_2)$. 



Given the above lemma we use a numerical optimizer
to compute the value $\alpha(c, f)$ for a particular choice of the parameters
$c$ and $f$. 
The result is the following claim. 

\begin{numericalclaim}
  \label{claim:pairing parameters}
  For $c = 0.8056$ and $f(x) = 1.618 \max(0, x-0.478)$, $\alpha(c, f)
  \ge \alphasecondfull$.  
\end{numericalclaim}

For our formal proof that we can achieve approximation ratio at least
$\alphasecond$, we need to modify Algorithm~\ref{alg:generic} slightly
to exlude certain types of configurations that are challenging for our
prover program.  In particular, we are only able to prove a good
approximation ratio for configurations in which all $|\mu_i|$'s and
$|\rho_{ij}|$'s are bounded away from $1$, so we modify
Algorithm~\ref{alg:generic} to perform a simple preprocessing step on
the vectors first to make sure that they are not too close to being
integral.  The details of this appears in Section~\ref{sec:num-proofs}
with the $\alphasecond$-algorithm being given by
Theorem~\ref{thm:rig-pairing}.  The choice of $c$ and $f$ used in our
formal proof is the same as in Claim~\ref{claim:pairing parameters}.

Figure~\ref{fig:fplot} shows the graphs of the two functions $\mu
\mapsto c \mu$ and $\mu \mapsto c\mu + (1-c)f(\mu)$, corresponding to
the typical lower and upper bound for the bias $r$ as a function of
$\mu$, for the values used in Claim~\ref{claim:pairing parameters}.

\begin{figure}
  \centering
  \scalebox{1}{\import{plots/}{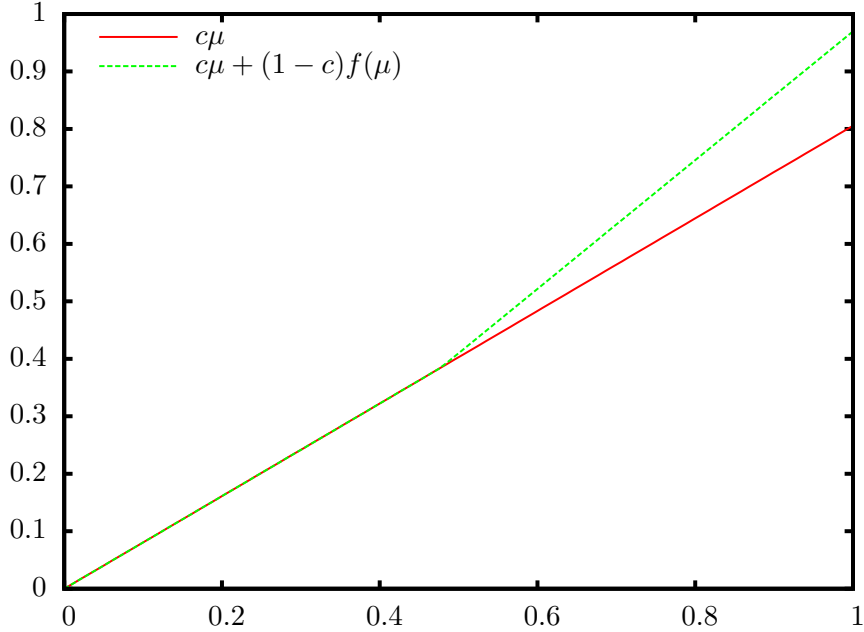}}
  \caption{The two functions $c\mu$ and $c\mu + (1-c)f(\mu)$.}
  \label{fig:fplot}
\end{figure}

When attempting to improve the approximation ratio, it turns out that
there are now several different forms of critical or near-critical
configurations, each of which imposes some restrictions on the
behaviour of $c$ and $f$.  Moreover, as is common for this type of
algorithm, our computations indicate that the worst configurations
$\mu_1, \mu_2, \rho$ lie at the surface of the space of
configurations $\Conf$.  In Figure~\ref{fig:surface ratio}, we give
contour plots of $\alpha_{c,f}(\mu_1, \mu_2, \rho)$ along this surface.

\begin{figure}
  \centering
  \subfigure[Lower envelope $\rho = -1 + |\mu_1 + \mu_2|$]{\scalebox{1}{\import{plots/}{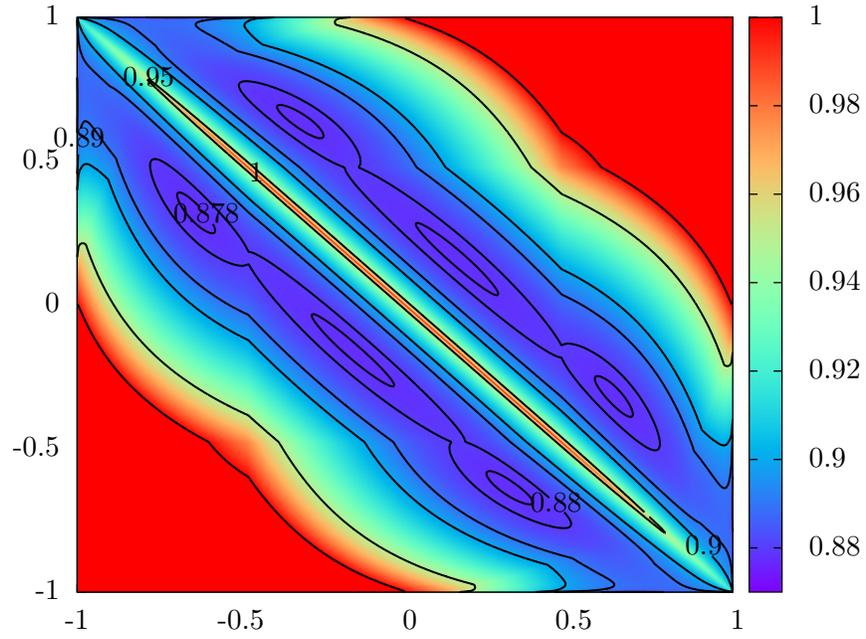}}}
  \subfigure[Upper envelope $\rho = 1 - |\mu_1 - \mu_2|$]{\scalebox{1}{\import{plots/}{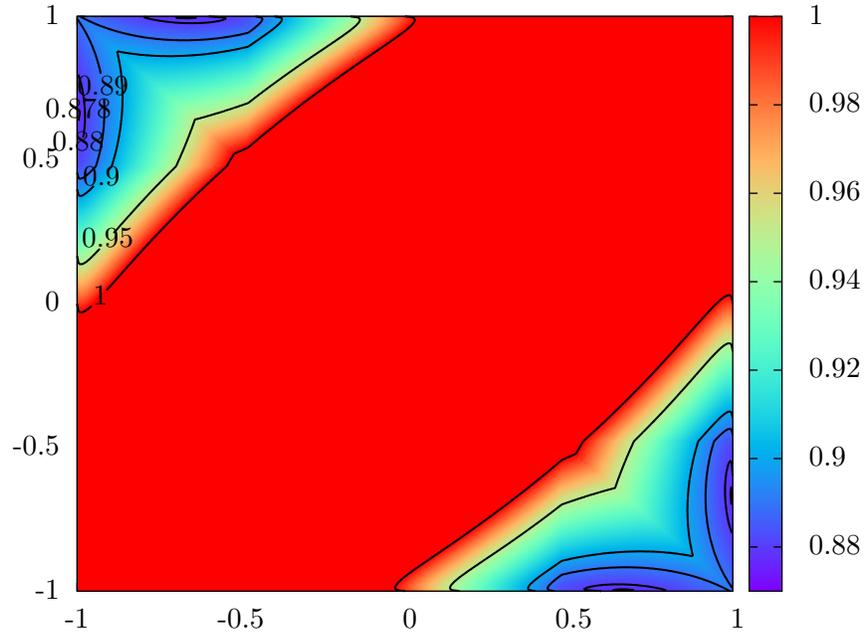}}}
  \caption{Approximation ratio along the surface of $\Conf$ when projected on the $(\mu_1,\mu_2)$-plane.}
  \label{fig:surface ratio}
\end{figure}




We now come back to the proof of Lemma~\ref{lemma:worstr} restated here for convenience.

\begin{relemma}{lemma:worstr} 
  For every configuration $\mu_1, \mu_2, \rho$ and closed intervals
  $I_1 = [a_1, b_1], I_2 = [a_2, b_2]$, we have that 
  \[\alpha(\mu_1, \mu_2, \rho, I_1, I_2) = \min_{(r_1, r_2) \in S \cap I_1 \times I_2} \alpha(\mu_1, \mu_2, \rho, r_1, r_2),\]
  where $S$ is defined as follows. Recall that $\trho = \frac{\rho-\mu_1\mu_2}{\sqrt{(1-\mu_1^2)(1-\mu_2^2)}}$.
  \begin{itemize}
  \item If $\trho \le 0$ then $S$ is the extreme points of the set $I_1 \times I_2$, i.e.
    \[S = \set{(a_1, a_2), (a_1, b_2), (b_1, a_2), (b_1, b_2)}.\]
  \item If $\trho > 0$ then $S$ is the extreme points plus five extra points defined in terms of the function
    $g(x) = 1 - 2 \Phi\left(\Phi^{-1}\left(\frac{1-x}{2}\right) / \trho\right)$. More precisely,
    \[S = \{(0, 0), (a_1, a_2), (a_1, b_2), (b_1, a_2), (b_1, b_2), (a_1, g(a_1)), (b_1, g(b_1)), (g(a_2), a_2), (g(b_2), b_2)\}.\]
  \end{itemize}
\end{relemma}

\begin{proof}
  We consider several cases depending on the value of $\trho$.
  \paragraph{Case 1: $|\trho| < 1$.} Using
  Lemma~\ref{lemma:gamma deriv} and the definition of
  $\Lambda_{\trho}$ we have that
  \begin{align*}
    \frac{\deriv}{\deriv r_1} \Lambda_{\trho}(r_1, r_2) &= \frac{1}{2} - \Phi\left(\frac{t(r_2) - \trho t(r_1)}{\sqrt{1-\trho^2}}\right),
  \end{align*}
  where we write $t(r) = \Phi^{-1}\left(\frac{1-r}{2}\right)$.  Thus,
  \begin{equation}
    \label{eqn:alpha deriv}
    \frac{\deriv}{\deriv r_1}\alpha(\mu_1, \mu_2, \rho,
    r_1, r_2) = \frac{-2}{1-\rho} \frac{\deriv \Lambda_{\trho}}{\deriv
      r_1}(r_1, r_2) = \frac{1}{1-\rho}\left(2\Phi\left(\frac{t(r_2) - \trho t(r_1)}{\sqrt{1-\trho^2}}\right)-1\right).
  \end{equation}

  \paragraph{Subcase 1.1: $-1 < \trho \le 0$.}
  Computing the second derivative of $\alpha$ with respect to $r_1$ we
  have
  \begin{align*}
    \frac{\deriv^2}{\deriv r_1^2}\alpha(\mu_1, \mu_2, \rho, r_1, r_2)
    &=-\frac{2\trho t'(r_1)}{(1-\rho)\sqrt{1-\trho^2}}\phi\left(\frac{t(r_2) - \trho t(r_1)}{\sqrt{1-\trho^2}}\right)\\
    &=\frac{\trho}{\phi(t(r_1))(1-\rho)\sqrt{1-\trho^2}}\phi\left(\frac{t(r_2) - \trho t(r_1)}{\sqrt{1-\trho^2}}\right) \leq 0.
  \end{align*}
  where we have used that $\frac{\deriv}{\deriv x}\Phi^{-1}(x) = 1/\phi(\Phi^{-1}(x))$.
  Thus $\alpha$ is concave in $r_1$ in this subcase.  By
  symmetry the same holds for $r_2$ as well which implies that 
  $\alpha(\mu_1, \mu_2, \rho, I_1, I_2)$ is minimized at one of the
  four extreme points as claimed.

  \paragraph{Subcase 1.2: $0 < \trho < 1$.}
  This case requires a little more work.
  Fix some minimum $r_1^*$, $r_2^*$ of \eqref{eq:worstr}.  If $(r_1^*,
  r_2^*) \in \{a_1, b_1\} \times \{a_2, b_2\}$ we are done, so we can
  assume that one of $r_1^*$ and $r_2^*$ lies strictly inside its interval.
  Suppose $r_1^*$ is in the interior of $I_1$, i.e., $a_1 < r_1^* < b_1$.  Then
  necessarily 
  \begin{equation*}
    \frac{\deriv}{\deriv r_1}\alpha(\mu_1, \mu_2, \rho,
  r_1^*, r_2^*) = 0.
  \end{equation*}
  By \eqref{eqn:alpha deriv} this implies that
  \[\Phi\left(\frac{t(r_2^*) - \trho t(r_1^*)}{\sqrt{1-\trho^2}}\right) = 1/2,\]
  which has the unique solution
  \begin{align}
    \frac{t(r_2^*) - \trho t(r_1^*)}{\sqrt{1-\trho^2}} &= 0, \nonumber\\
    \intertext{or equivalently,}
    \label{eqn:alpha deriv zero cond}
    t(r_2^*) &= \trho t(r_1^*),
  \end{align}
  Similarly, if $a_2 < r_2^* < b_2$, it must be the case that
  $t(r_1^*) = \trho t(r_2^*)$.
  
  This implies that if \emph{both} $r_1^*$ and $r_2^*$ lie strictly inside
  their respective intervals then $t(r_1^*) = \trho t(r_2^*) =
  \trho^2 t(r_1^*)$. As $|\trho| < 1$ this implies $t(r_1^*) = t(r_2)^* = 0$ 
  which has the unique solution $r_1^* = r_2^* = 0$.

  On the other hand if \emph{exactly one} of $r_1^*$ and $r_2^*$ lies
  strictly inside its respective interval, say $r_1^*$, then by
  \eqref{eqn:alpha deriv zero cond}, $r_1^* = t^{-1}(t(r_2^*)/\trho)=g(r_2^*)$.  

  \paragraph{Case 2: $\trho = 1$.}  In this case, 
  \begin{equation}
    \label{eqn:lambda when rho is 1}
    \alpha(\mu_1, \mu_2, \rho, r_1, r_2) = \frac{2\left(1 - \Lambda_1(r_1, r_2) \right)}{1 - \rho}
    = \frac{2\left(1 - \pro[g\sim \mathcal{N}(0, 1)]{\strut g \not\in [t(r_1), t(r_2)]} \right)}{1 - \rho} = \frac{|r_1-r_2|}{(1-\rho)}.
  \end{equation}
  The minimizer $(r_1^*, r_2^*)$ of this expression depends on
  whether $I_1 \cap I_2 = \emptyset$ or not.  If $I_1 \cap I_2 =
  \emptyset$ then the unique minimizer $(r_1^*, r_2^*)$ is in $\{a_1,
  b_1\} \times \{a_2, b_2\}$.  Otherwise, if $I_1 \cap I_2 \ne
  \emptyset$, then the minimum is zero and any $r_1^* = r_2^* = r^* \in I_1 \cap I_2$ is a
  minimizer of \eqref{eqn:lambda when rho is 1}. In particular we
  can choose $r^*$ to be the endpoint of one of the intervals.  Noting
  that when $\trho = 1$ we have $g(x) = x$ finishes this case.

  \paragraph{Case 3: $\trho = -1$.} 
  Similarly to the previous case we now have
  \[
  \alpha(\mu_1, \mu_2, \rho, r_1, r_2) = \frac{2\left(1 - \Lambda_{-1}(r_1, r_2) \right)}{1 - \rho}
    = \frac{2\left(1 - \pro[g\sim \mathcal{N}(0, 1)]{\strut g \in [-t(r_1), t(r_2)]} \right)}{1 - \rho} = \frac{2-|r_1+r_2|}{(1-\rho)}.
  \]
  The unique minimizer $(r_1^*, r_2^*)$ of this expression is clearly in $\{a_1, b_1\} \times \{a_2, b_2\}$.
\end{proof}


\section{\balmaxtwosat{}}\label{sec:2sat}
Algorithm~\ref{alg:generic} can be directly applied to any \balcsp{P}.
In particular it is interesting to do this for \balmaxtwosat{}.

In the language of this paper, the best algorithm for \maxtwosat{}
\cite{lewin02improved} uses a linear bias selection algorithm
$r_i = c \cdot b_i$ (see description in \cite{austrin07balanced}) and
so it already satisfies the constraint $\sum r_i = 0$.  Thus it
immediately extends to the case of \balmaxtwosat{}, implying that this
problem is approximable within $\alpha_{LLZ} - \epsilon$ for every
$\epsilon > 0$, where $\alpha_{LLZ} \approx 0.9401$ is the
approximation threshold for \maxtwosat{} assuming the UGC.

Furthermore, it is easy to see that for any predicate $P$, \balcsp{P}
is at least as hard as \csp{P};  the reduction from \csp{P} to
\balcsp{P} simply produces two disjoint copies of the \csp{P} instance
and negates all literals in one of the copies.

In particular, this implies that assuming the UGC, the approximation
threshold of \balmaxtwosat{} is the same as the threshold for
\maxtwosat{}, namely $\alpha_{LLZ} \approx 0.9401$.  This fact, that
the balance constraint does not make \maxtwosat{} harder, can be seen
as circumstantial evidence that \maxbisection{} is as easy as
\maxcut{}.


\section{Proofs of Approximation Ratios}\label{sec:num-proofs}

Unfortunately, our formal proofs of approximation ratios are based on
case analysis of several million cases, and we therefore have to
construct them with the assistance of a computer.  The case analysis
is similar to that of e.g., \cite{zwick02computer,sjogren09rigorous}
and proceeds by recursively dividing the search space $[-1,1]^3$ into
subcubes.  When processing a cube $C \subseteq [-1,1]^3$, we can compute lower and upper
bounds on the performance of our algorithm $\alpha(\mu_1, \mu_2, \rho, r_1, r_2)$ for $(\mu_1, \mu_2,
\rho) \in C$.  To handle this and to also take care of the rounding
errors inherent in finite precision calculations, we use interval
arithmetic.

When processing a cube $C$, there are four possibilies:
\begin{enumerate}
\item $C$ is completely outside the space of configurations $\Conf$.  
\item The lower bound on $\alpha$ in the cube exceeds the
  approximation ratio we are trying to prove.  
\item The upper bound on $\alpha$ in the cube is lower than the
  approximation ratio we are trying to prove.  
\item None of the above: the case is inconclusive.  Then we subdivide
  $C$ into eight subcubes in the natural way, and we check each of them recursively.  
\end{enumerate}
Note that we need to run the above test till we reach our precision threshold and no inconclusive cases remain. Also, this will translate into a proof for our approximation performance as long as we avoid case 3. 
Unfortunately, it turns out that there is one issue to deal
with.  Specifically, consider a configuration $(\mu_1, \mu_2, \rho)$
where $\mu_1 \approx \pm 1$, or more precisely a cube $C$ such that
all configurations in $C$ have $\mu_1 \approx \pm 1$.  Then the
dependence of $\trho =
\frac{\rho-\mu_1\mu_2}{\sqrt{1-\mu_1^2}\sqrt{1-\mu_2^2}}$ on $\mu_1$
is not Lipschitz continuous meaning that even when the cube is small
the uncertainty in $\trho$ can be very large, which in turn results in
poor bounds on the value of $\alpha$ and in particular our lower bound
will not be strong enough to conclude that this case is not problematic.
This turns out to be not just a hypothetical issue, but a very real one,
because in our algorithm there are worst or near-worst configurations
which have $\mu_1$ (or $\mu_2$) close or equal to $\pm 1$.  A similar
issue occurs when $\rho \approx 1$, in which case the SDP value is
close to $0$ and we need very sharp estimates on $\trho$ in order to
get a sufficiently strong lower bound on $\alpha$.  To overcome this,
the simplest recourse is to slightly alter Algorithm~\ref{alg:generic} 
by adding a preprocessing step which precludes
configurations $(\mu_1,\mu_2,\rho_{12})$ where $|\mu_i|$ or $|\rho|$ is close to $1$.  This causes us an
additional small loss in the SDP objective value.  We have the
following theorem.

\begin{lemma}
  \label{lemma:make sdp solution smooth}
  Given $\delta > 0$ and an SDP solution $\bv_0, \ldots, \bv_n$ (unit
  vectors), we can construct in polynomial time a new SDP solution
  $\bv_0', \ldots, \bv_n'$ (unit vectors) such that
  \begin{enumerate}
  \item $\SDPVal(\{\bv'_i\}) \ge \SDPVal(\{\bv_i\}) - \delta$,
  \item $|\iprod{\bv'_i}{\bv'_j}| \le 1 - \delta$ for every $0 \le i < j\le k$,
  \item If $\{\bv_i\}$ satisfies the triangle inequalities than so
    does $\{\bv_i'\}$,
  \item If $\{\bv_i\}$ is $\epsilon$-uncorrelated for some $\epsilon \ge
    0$ then so is $\{\bv_i'\}$.
  \end{enumerate}
\end{lemma}

\begin{proof}[Proof sketch]
  We replace every vector $\bv_1, \ldots, \bv_n$ by $\bv_i' =
  \sqrt{1-\delta} \bv_i + \sqrt{\delta} \bu_i$, where $\bu_i$ is a
  unit vector orthogonal to all other vectors (we keep $\bv'_0 =
  \bv_0$ the same).  This has the effect of scaling all $\mu_i$'s by $1-\delta$ and all $\rho_{ij}$'s by $(1-\delta)^2$, i.e., 
  \begin{align*}
    \mu_i' &= (1-\delta) \mu_i & \rho'_{ij} &= (1-\delta)^2 \rho_{ij}.
  \end{align*}
As a result, 
the four items can be proven through straightforward
  calculations.  The last item may be easier to think about in the
  probabilistic view: in terms of the local distributions, the
  transformation we did has the effect of mixing the local
  distributions with the uniform distribution, which clearly only
  decreases correlations.
\end{proof}

With this lemma in place, it is natural to introduce a variation
$\Conf_\delta$ of the space of configurations $\Conf \subseteq
[-1,1]^3$, where we exclude all configurations where some coordinate
exceeds $1-\delta$ in absolute value, i.e., 
\[
\Conf_\delta := \Conf \cap [-1+\delta, 1-\delta]^3.
\]
We refer to such configurations as \emph{smooth}.  We then extend the
various $\alpha$ definitions which involve minimization over $\Conf$
in a similar way: analogously to Definition~\ref{def:linear_alpha} we
write
\[\alpha_\delta(c) := \min_{(\mu_1, \mu_2, \rho) \in \Conf_\delta} \alpha(\mu_1, \mu_2, \rho, c \cdot \mu_1, c \cdot \mu_2),\]
and analogously to Definition~\ref{def:pairing
  alpha} we write
\[\alpha_\delta(c, f) = \min_{(\mu_1, \mu_2, \rho) \in \Conf_\delta} \alpha_{c,f}(\mu_1, \mu_2, \rho).\]

By Lemma~\ref{lemma:make sdp solution smooth}, if $\alpha_\delta(c,f)
\ge \alpha$ for some $c$, $f$ and $\delta$, using the framework of
Section~\ref{sec:general} we immediately obtain an $(\alpha - \delta -
\epsilon)$-approximation algorithm (for any $\epsilon > 0$, with
running time $O(n^{\poly(1/\epsilon)})$), and similarly for
$\alpha_\delta(c)$.

By computer-assisted case analysis, we are able to prove the following
two theorems, lower bounding the approximation ratios of our two types
of rounding on smooth configurations. First, we are able to justify the performance of our first algorithm as presented in Section~\ref{sec:first_improvement}. 
\begin{theorem}\label{thm:rig-linear}
  For $c = 0.86451$ and $\delta = 10^{-5}$, we have $\alpha_\delta(c) \ge
  0.87362$.
\end{theorem}
The proof of Theorem~\ref{thm:rig-linear} consists of roughly $20$
million cases and the theorem takes about $9$ minutes to prove on a
SunFire X2270 machine with Intel X5675 CPUs.

Most importantly, the next theorem implies our improved approximation guarantee, as described in 
Theorem~\ref{thm:main}.
\begin{theorem}\label{thm:rig-pairing}
  For $c = 0.8056$, $f(x) = 1.618 \max(0, x-0.478)$ and $\delta =
  10^{-5}$, we have $\alpha_{\delta}(c,f) \ge 0.87762$.
\end{theorem}
The proof of Theorem~\ref{thm:rig-pairing} consists of roughly $140$
million cases and the theorem takes about $25$ minutes to prove on a
SunFire X2270 machine with Intel X5675 CPUs.

\section{Conclusion and Future Work}
\label{sec:conclusions}

We introduced a new class of rounding algorithms for the
\maxbisection{} problem and \balcsps{} extending the work of
\cite{raghavendra12approximating}. We analyzed the results to present
a \alphasecond{}-approximation algorithm for \maxbisection{} and an
$(\alpha_{LLZ}-\epsilon)$-approximation algorithm for \balmaxtwosat{},
improving on approximation ratios of $\RTratio$ and $0.93$
respectively \cite{raghavendra12approximating}.  Our improved bound $\alphabest$ is so far based on
extensive numerical evidence, but we are currently working on a formal
proof of this bound.  Our algorithm for \balmaxtwosat{} is optimal
assuming the \ugc{} and the ratio of our algorithm for \maxbisection{} is off
from the UGC-hardness threshold by less than $10^{-3}$.
The most obvious open question is to
close this small gap. We conjecture that there is an
$(\alphagw-\epsilon)$-approximation algorithm for \maxbisection{},
i.e., it has the same approximation threshold as \maxcut{}.

It is worth noting that there are constraint satisfaction problems
where adding a bisection constraint makes the problem strictly harder.
A natural example is \textsc{Min Cut} which is
solvable in polynomial time but its bisection variant, \textsc{Min
  Bisection}, is NP-hard to solve exactly and R3SAT-hard to
approximate within a factor $4/3$ \cite{feige02relations}.

It would be interesting to come up with a generic algorithm family that provides the best approximation algorithm for all \balcsp{P} problems. In particular, while the seminal work of Raghavendra~\cite{raghavendra08optimal} shows that, assuming the UGC, for any predicate $P$ the best approximation algorithm for \csp{P} is to run a certain rounding scheme on its natural SDP relaxation there is no analog for \balcsp{P}. Notice that \cite{raghavendra08optimal} \emph{does not} provide a (practical) way to compute the approximation factor of this algorithm and just proves its optimality, hence a parallel result for \balcsp{P} would be incomparable to the current paper and \cite{raghavendra12approximating}.


\bibliographystyle{alpha}
\bibliography{references}
\begin{appendix}
\section{Proofs of Some Properties of the Bivariate Gaussian Distribution}
\label{sec:gaussian proofs}

In this section we prove some lemmata from Section~\ref{sec:gaussian prelims}

\begin{relemma}{lemma:gammasymmetry}
  For every $\trho \in [-1,1]$, $q_1, q_2 \in [0,1]$, we have
  \[
  \Gamma_{\trho}(1-q_1, 1-q_2) = \Gamma_{\trho}(q_1, q_2) + 1 - q_1 - q_2.
  \]
\end{relemma}
\vspace{-1cm}
\begin{proof}
  Define $(X, Y)$ as a pair of jointly Gaussian random variables each of which has mean 0 and variance 1, where $Cov(X, Y) = \trho$.
  From definition of $\Gamma$ we have,
  \begin{align*}
    \Gamma_{\trho}(1-q_1, 1-q_2) &= \pro{X \leq \Phi^{-1}(1-q_1) \wedge Y \leq \Phi^{-1}(1-q_2)}\\
    &= \pro{-X \geq -\Phi^{-1}(1-q_1) \wedge -Y \geq -\Phi^{-1}(1-q_2)}.\\
    \intertext{Observing that $\Phi^{-1}(1-q_1)=-\Phi(q_1)$ and $(-X, -Y)$ has exactly the same distribution as $(X, Y)$, }
    &= \pro{X \geq \Phi^{-1}(q_1) \wedge Y \geq \Phi^{-1}(q_2)}
    = 1-\pro{X < \Phi^{-1}(q_1) \vee Y < \Phi^{-1}(q_2)}\\
    &= 1-\left(\pro{X < \Phi^{-1}(q_1)} + \pro{Y < \Phi^{-1}(q_2)} - \pro{X < \Phi^{-1}(q_1) \wedge Y < \Phi^{-1}(q_2)} \right)\\
    &= 1-q_1-q_2+\pro{X \leq \Phi^{-1}(q_1) \wedge Y \leq \Phi^{-1}(q_2)} = 1 - q_1 - q_2 + \Gamma_{\trho}(q_1, q_2),
  \end{align*}
  where we have used inclusion-exclusion and the fact that $\pro{X = \Phi^{-1}(q_1)} = 0$ in the last two lines of the proof.
\end{proof}

\begin{relemma}{lemma:gamma indep bound}
  For any $\trho \in [-1,1]$, $q_1, q_2 \in [0,1]$, we have
  \[\Gamma_{\trho}(q_1, q_2) \le q_1 q_2 + 2 |\trho|.\]
\end{relemma}
\vspace{-1cm}
\begin{proof}
  Let $t_1 = \Phi^{-1}(q_1)$, $t_2 = \Phi^{-1}(q_2)$.  We may assume $|\trho| \le 1/2$ since otherwise the Lemma is trivially true.  For any $x \in \R$, we have
  \begin{align*}
    \Pr[Y \le t_2 \,|\, X = x] &= \Phi\left(\frac{t_2 - \trho x}{\sqrt{1-\trho^2}} \right) \\
      &\le \Phi(t_2-\trho x) + \frac{1/\sqrt{1-\trho^2} - 1}{\sqrt{2 \pi e}} & \text{(by Lemma~\ref{lem:gaussian conc}, \eqref{eqn:gaussian anti-conc2})} \\
    &\le \Phi(t_2) + \frac{|\trho x|}{\sqrt{2 \pi}} + \frac{1/\sqrt{1-\trho^2} - 1}{\sqrt{2 \pi e}} & \text{(by Lemma~\ref{lem:gaussian conc}, \eqref{eqn:gaussian anti-conc1})}\\
      &\le \Phi(t_2) + \frac{|\trho x|}{\sqrt{2 \pi}} + \frac{\trho^2}{\sqrt{2 \pi e}}. & \text{(by $|\trho| \le 1/2$)}
  \end{align*}
  From this we conclude that
  \begin{align*}
    \Gamma_{\trho}(q_1, q_2) &= \int_{x=-\infty}^{t_1} \phi(x) \Phi\left(\frac{t_2 - \trho x}{\sqrt{1-\trho^2}}\right) \ud x \\
    &\le \int_{x = -\infty}^{t_1} \phi(x) \left(\Phi(t_2) + \frac{\trho^2}{\sqrt{2\pi e}} + \frac{|\trho x|}{\sqrt{2 \pi}}\right) \ud x \\
    & = q_1 q_2 + \frac{\trho^2}{\sqrt{2 \pi e}} + \frac{|\trho|}{\sqrt{2 \pi}} \int_{x = -\infty}^{t_1} |x| \phi(x) \ud x \\
    &\le q_1 q_2 + |\trho| + |\trho|,
  \end{align*}
  where in the last step we have used $\Exp{|X|} = \sqrt{2/\pi}$ for $X \sim \mathcal{N}(0,1)$. This completes the proof.
\end{proof}
The preceding proof used two standard anti-concentration bounds for Gaussians, as summarized by the following Lemma.

\begin{lemma}\label{lem:gaussian conc}
  Let $X\sim \mathcal{N}(0, 1)$ be a standard gaussian random variable, and $t \in \R$, $a\leq b$, $\alpha\geq 1$ real numbers. Then,
  \begin{align}
    \pro{a \leq X \leq b} &\leq (b-a)/\sqrt{2\pi}, \label{eqn:gaussian anti-conc1}\\
    \pro{X \leq \alpha t} &\leq \pro{X \leq t} + \frac{\alpha-1}{\sqrt{2\pi e}}. \label{eqn:gaussian anti-conc2}
  \end{align}
\end{lemma}
\begin{proof}
  To prove \eqref{eqn:gaussian anti-conc1} observe that
  \[\pro{a \leq X \leq b} = \frac{1}{\sqrt{2\pi}}\int_{x=a}^b e^{-x^2/2} \ud x \leq \frac{1}{\sqrt{2\pi}}\int_{x=a}^b \ud x = \frac{b-a}{\sqrt{2\pi}}.\]
  Proceeding to \eqref{eqn:gaussian anti-conc2}, the case $t < 0$ holds trivially.  For $t \ge 0$ we have
  \[\pro{t \leq X \leq \alpha t} = \frac{1}{\sqrt{2\pi}}\int_{x=t}^{\alpha t} e^{-x^2/2} \ud x \leq \frac{1}{\sqrt{2\pi}}\int_{x=t}^{\alpha t} e^{-t^2/2} \ud x = \frac{(\alpha-1) t e^{-t^2/2}}{\sqrt{2\pi}} \leq \frac{\alpha-1}{\sqrt{2 \pi e}},\]
  where the last inequality follows because the derivative of the function $f(t) = te^{-t^2/2}$ is $(1-t^2)e^{-t^2/2}$, hence
  $f(t)$ achieves its maximum at $t=1$.
\end{proof}

We now prove Lemma~\ref{lemma:gamma deriv} repeated here
for convenience.

\begin{relemma}{lemma:gamma deriv}
  For every $\trho \in (-1,1), q_1, q_2 \in [0,1]$, we have
  \begin{align*}
    \frac{\deriv}{\deriv q_1} \Gamma_{\trho}(q_1, q_2) &= \Phi\left(\frac{t_2 - \trho t_1}{\sqrt{1-\trho^2}}\right), 
  \end{align*}
  where $t_i = \Phi^{-1}(q_i)$.
\end{relemma}
\begin{proof}
  We have
\[\Gamma_{\trho}(q_1, q_2) = \int_{x=-\infty}^{t_1} \phi(x) \Phi\left(\frac{t_2 - \trho x}{\sqrt{1-\trho^2}}\right) \ud x,\]
giving
\begin{align*}
  \frac{\deriv}{\deriv q_1}\Gamma_{\trho}(q_1, q_2) &= \frac{\dif t_1}{\dif q_1} \phi(t_1) \Phi\left(\frac{t_2 - \trho t_1}{\sqrt{1-\trho^2}}\right) 
  = \Phi\left(\frac{t_2 - \trho t_1}{\sqrt{1-\trho^2}}\right),
\end{align*}
where the second step used $\frac{\dif}{\dif x}\Phi^{-1}(x) = \frac{1}{\phi(\Phi^{-1}(x))}$.
\end{proof}


\end{appendix}
\end{document}